\documentclass[11pt,aps,pra,notitlepage,nofootinbib,tightenlines]{revtex4-1}

\usepackage[T1]{fontenc}
\usepackage[english]{babel}
\usepackage[utf8]{inputenc}

\usepackage{amsmath}
\usepackage{amssymb}
\usepackage{amsthm}
\usepackage{mathtools}
\usepackage[dvipsnames]{xcolor}
\usepackage{tikz}
\usepackage{overpic}
\usepackage[colorlinks=true,linkcolor=blue,citecolor=blue]{hyperref}
\usepackage{pgfplots}
\usepackage{enumitem}
\setlist[enumerate]{noitemsep,partopsep=0pt,parsep=0pt}
\setlist[itemize]{noitemsep,partopsep=0pt,parsep=0pt}

\newtheorem{thm}{Theorem}

\newtheorem{cor}[thm]{Corollary}
\newtheorem{prop}[thm]{Proposition}
\newtheorem{defn}[thm]{Definition}

\theoremstyle{definition}
\newtheorem{ex}[thm]{Example}

\theoremstyle{remark}
\newtheorem{rem}{Remark}

\numberwithin{equation}{section}

\DeclareMathOperator{\CPTP}{CPTP}
\DeclareMathOperator{\spec}{spec}

\DeclareMathOperator{\supp}{supp}

\DeclareMathOperator{\rank}{rank}
\DeclareMathOperator{\im}{im}
\DeclareMathOperator{\tr}{tr}
\DeclareMathOperator{\Proj}{Proj}
\DeclareMathOperator{\sunit}{S_{unit}}
\DeclareMathOperator{\scanonical}{S_{can}}
\DeclareMathOperator{\sinvariant}{S_{inv}}
\newcommand{\ket}[1]{|#1\rangle}
\newcommand{\bra}[1]{\langle#1|}
\newcommand{\ketbra}[2]{\ket{#1}\!\bra{#2}}
\newcommand{\proj}[1]{\ketbra{#1}{#1}}
\newcommand{\braket}[2]{\langle#1|#2\rangle}
\newcommand{\brak}[1]{\braket{#1}{#1}}
\newcommand{\swapf}{\mathcal{F}}
\newcommand{\swap}[2]{\mathcal{F}_{#1#2}}

\allowdisplaybreaks

\begin{document}

\title{Min-reflected entropy = doubly minimized Petz R\'enyi mutual information of order 1/2}
\author{Laura Burri}
\affiliation{Institute for Theoretical Physics, ETH Zurich, Zurich, Switzerland}

\begin{abstract}
R\'enyi reflected entropies of order $n\geq 2$ are correlation measures that have been introduced in the field of holography. 
In this work, we put the spotlight on the min-reflected entropy, i.e., the R\'enyi reflected entropy in the limit $n\rightarrow\infty$. 
We show that, for general bipartite quantum states, this measure is identical to another measure originating from the field of quantum information theory: the doubly minimized Petz R\'enyi mutual information of order $1/2$. 
Furthermore, we demonstrate how this equality enables us to answer several previously open questions, each concerning one of the two correlation measures (or generalizations of them).
\end{abstract}

\maketitle

\section{Introduction}\label{sec:introduction}
Quantum correlations between two systems $A$ and $B$, as described by a bipartite quantum state $\rho_{AB}$, can occur in various forms. 
The usefulness of a given state $\rho_{AB}$ for a particular information-theoretic task depends on the task's specific requirements 
and is often quantified by a \emph{correlation measure}, such as the mutual information or one of its R\'enyi generalizations.

As the main technical contribution of this work, we prove that two correlation measures that have been studied independently in previous work are, in fact, identical: 
The min-reflected entropy is equal to the doubly minimized Petz R\'enyi mutual information of order $1/2$,
\begin{align}\label{eq:main}
S_R^{(\infty)}(A:B)_\rho = I_{1/2}^{\downarrow\downarrow}(A:B)_{\rho} .
\end{align}
The measure on the left-hand side has been introduced and primarily examined in the research field of holography, while the measure on the right-hand side is well-known in the research field of quantum information theory. 
Below, we provide a brief overview of the contexts in which these two measures have been studied in previous work and outline some unresolved questions, each concerning one of the two measures (or generalizations of them). 
In subsequent sections, we will demonstrate how our main result in~\eqref{eq:main} enables us to answer these questions, showing that it is not only of formal interest but also has insightful implications.

\textbf{The left-hand side.} 
The left-hand side of~\eqref{eq:main} belongs to the family of \emph{R\'enyi reflected entropies}. 
This family is defined for two parameters $m\in [0,\infty),n\in [0,\infty]$ as~\cite{dutta2019canonical}
\begin{equation}\label{eq:sr-mn}
S_R^{(m,n)}(A:B)_\rho \coloneqq H_n (AA^*)_{\hat{\rho}^{(m)}}
\end{equation}
where the right-hand side is the R\'enyi entropy of order $n$ of the marginal state on $AA^*$ of $\ket{\hat{\rho}^{(m)}}_{ABA^*B^*}$, which denotes the canonical purification of $\rho_{AB}^m/\tr[\rho_{AB}^m]$. 
Even though the R\'enyi reflected entropy is most commonly defined for two parameters $(m,n)$, it is sometimes convenient to define the R\'enyi reflected entropy with a single parameter~\cite{hayden2023reflected} by setting $m=1$, i.e.,  $S_R^{(n)}(A:B)_\rho\coloneqq S_R^{(1,n)}(A:B)_\rho$. 
The case where $(m,n)=(1,1)$ is known as the \emph{reflected entropy} $S_R(A:B)_\rho$. 
To date, this member of the family has received the most attention in the literature. 
In contrast, our equality in~\eqref{eq:main} highlights the member with $(m,n)=(1,\infty)$, which we call the \emph{min-reflected entropy}.

The reflected entropy was originally introduced in~\cite{dutta2019canonical}. 
They argued that, in holographic theories, half the reflected entropy between two boundary regions is dual to the area of the \emph{entanglement wedge cross section} in the bulk. 
Other proposals for the boundary dual of the entanglement wedge cross section include the entanglement of purification~\cite{umemoto2018entanglement,nguyen2018entanglement,bhattacharyya2018entanglement,hirai2018towards,bao2018holographic}, 
the $R$-correlation~\cite{umemoto2019quantum,levin2020correlation}, 
the logarithmic negativity~\cite{kudlerflam2019entanglement,kusuki2019derivation}, 
the odd entropy~\cite{tamaoka2019}, 
and the entanglement of formation~\cite{mori2024doesconnectedwedgeimply}. 
Not only the reflected entropy but the whole family of measures in~\eqref{eq:sr-mn} was originally introduced in~\cite{dutta2019canonical}. 
The rationale behind introducing this family in~\cite{dutta2019canonical} was to provide a computational tool for calculating the reflected entropy: 
By first computing $S_R^{(m,n)}(A:B)_\rho$ via the replica trick for some natural numbers $m,n\geq 2$, and then performing an analytic continuation in $m$ and $n$, the reflected entropy can be obtained by taking the limit $m,n\rightarrow 1$~\cite{dutta2019canonical}.

Since its introduction, the reflected entropy has gained attention not only in holography~\cite{jeong2019reflected,akers2020entanglement,kudlerflam2020correlation,
marolf2020cft,chandrasekaran2020including,moosa2020time,kusuki2021dynamics,
afrasiar2022covariant,li2022defect,akers2022page,chen2022reflected,basu2022entanglement,
ling2022reflected,afrasiar2023reflected,afrasiar2023reflected2,basak2023holographic,vasli2023holographic,
basu2024reflected,basu2024entanglementttrotatingblack,ahn2024renyireflectedentropyentanglement}, 
but also in other fields such as 
conformal field theory~\cite{camargo2021long,kudlerflam2021quasi,berthiere2023reflected,jiang2024entanglementmembrane2dcft},
condensed matter physics~\cite{berthiere2021topological,sohal2023entanglement,liu2024multipartite}, 
free scalar and fermionic field theory~\cite{bueno2020reflected,bueno2020reflected2,basak2023reflectedentropymarkovgap,dutta2023reflected}, 
and random tensor networks~\cite{akers2022reflected,akers2023reflected,akers2024reflectedentropyrandomtensor}. 
Many of these works~\cite{moosa2020time,kudlerflam2021quasi,berthiere2021topological,
akers2022reflected,akers2022page,basu2022entanglement,chen2022reflected,
akers2023reflected,afrasiar2023reflected,afrasiar2023reflected2,
akers2024reflectedentropyrandomtensor,basu2024reflected,basu2024entanglementttrotatingblack,ahn2024renyireflectedentropyentanglement} 
use the replica trick as a practical tool for calculating the reflected entropy.

The reflected entropy and its R\'enyi variants are often interpreted as correlation measures~\cite{dutta2019canonical,kudlerflam2020correlation,hayden2021markov,kusuki2021dynamics,akers2022reflected,akers2023reflected,
kudlerflam2021quasi,kusuki2020entanglement,basu2022entanglement,liu2022multipartitioning,
liu2024multipartite,basu2024reflected,basu2024entanglementttrotatingblack}. 
The R\'enyi reflected entropy of order $n\in \mathbb{N}_{\geq 2}$ is known to be monotonic under the inclusion of subsystems~\cite{dutta2019canonical}, i.e., $S_R^{(n)}(A:B)_\rho\leq S_R^{(n)}(A:BC)_\rho$ for any tripartite quantum state $\rho_{ABC}$, which implies that these measures are indeed correlation measures. 
This monotonicity, however, can be violated for $n\in (0,2)$, which means that the corresponding R\'enyi reflected entropies do not qualify as correlation measures~\cite{hayden2023reflected}. 
Since this range includes the reflected entropy ($n=1$), it follows that the reflected entropy is not a correlation measure for general quantum states~\cite{hayden2023reflected}, 
even though it appears to behave as a correlation measure for certain special types of states such as holographic states~\cite{dutta2019canonical} or certain free scalar fields~\cite{bueno2020reflected}. 
Nevertheless, it remains possible that the reflected entropy accurately indicates maximal entanglement for general quantum states. 
We thus ask:

\begin{description}
    \item[Question~1] Does the maximality of the reflected entropy (or any other R\'enyi reflected entropy) correspond to maximal entanglement? 
    More precisely, does the R\'enyi reflected entropy achieve its maximal value of $2\log\min \{\dim(A),\dim(B)\}$ iff $A$ and $B$ are maximally entangled?
\end{description}

The aforementioned results offer insights into whether R\'enyi reflected entropies are correlation measures. 
However, a more profound question remains: 
Do they have operational relevance, in the sense of quantifying some aspect of an information-theoretic task? 
Partial progress toward an information-theoretic understanding of the reflected entropy has been made in~\cite{akers2020entanglement,zou2021universal,hayden2021markov}. 
These works examine the difference between the reflected entropy and the mutual information --- named the \textit{Markov gap (of reflected entropy)}~\cite{hayden2021markov}. 
According to these works, the Markov gap serves as a probe of tripartite entanglement~\cite{akers2020entanglement,zou2021universal}, 
and it has been shown~\cite{hayden2021markov} that the Markov gap is lower-bounded by a function of the fidelity of a particular Markov recovery process. 
However, since the Markov gap is the \emph{difference} between the reflected entropy and the mutual information, the aforementioned works attribute only some information-theoretic significance to this difference, but not to the reflected entropy itself. 
The question of whether the reflected entropy (or any of its R\'enyi variants) has an operational meaning thus remains largely unresolved. 
It has been conjectured~\cite{akers2022page} that a deeper connection exists between the reflected entropy and information recovery via the Petz map, but no general arguments supporting this conjecture have been presented to date. 
We thus ask:

\begin{description}
    \item[Question~2] Is there a member of the family of R\'enyi reflected entropies that has an operational interpretation? 
    In particular, is there a general relationship between a R\'enyi reflected entropy and information recovery via the Petz map?
\end{description}

\textbf{The right-hand side.}
The right-hand side of~\eqref{eq:main} belongs to the family of the \emph{doubly minimized Petz R\'enyi mutual information (PRMI) of order $\alpha$}. This family is defined for $\alpha\in [0,\infty)$ as~\cite{burri2024doublyminimizedpetzrenyi}
\begin{equation}\label{eq:i-downdown}
I_\alpha^{\downarrow\downarrow}(A:B)_\rho
\coloneqq\inf_{\substack{\sigma_A,\tau_B}}D_\alpha(\rho_{AB}\| \sigma_A\otimes \tau_B)
\end{equation}
where $D_\alpha$ denotes the Petz divergence and the minimization is over all quantum states $\sigma_A,\tau_B$ on $A,B$. 
The case $\alpha=1$ coincides with the mutual information $I(A:B)_\rho$~\cite{gupta2014multiplicativity,hayashi2016correlation,burri2024doublyminimizedpetzrenyi}. 
The case $\alpha=1/2$ is the measure that appears on the right-hand side of~\eqref{eq:main}. 
General properties of the doubly minimized PRMI of order $\alpha$ have been studied in~\cite{burri2024doublyminimizedpetzrenyi}. 
Although several properties, such as additivity for $\alpha\in [1/2,2]$, are known, no closed-form expression for the doubly minimized PRMI of order $\alpha$ has yet been found (except for $\alpha=1$)~\cite{burri2024doublyminimizedpetzrenyi}.

The classical analogue of the doubly minimized PRMI has been studied in~\cite{tomamichel2018operational,lapidoth2018testing,lapidoth2019two}. 
It is defined for a probability mass function (PMF) $P_{XY}$ as $I_\alpha^{\downarrow\downarrow}(X:Y)_P\coloneqq \inf_{Q_X,R_Y}D_\alpha (P_{XY}\| Q_XR_Y)$ where $D_\alpha$ denotes the R\'enyi divergence of order $\alpha\in [0,\infty)$ and the minimization is over PMFs $Q_X,R_Y$ of the random variables $X,Y$. 
Also for this measure, no closed-form expression has been found yet (except for $\alpha=1$)~\cite{lapidoth2019two}. 
However, for $\alpha=1/2$, at least a \emph{largest eigenvalue formulation} has been derived~\cite[Lemma~6]{lapidoth2019two}, which is typically easier to evaluate than the definition-based expression. 

The classical analogue of the doubly minimized PRMI can be interpreted as one of its special cases: 
If the doubly minimized PRMI $I_\alpha^{\downarrow\downarrow}(A:B)_\rho$ is evaluated for a classical-classical (CC) state $\rho_{AB}$ with PMF $P_{XY}$, then it coincides with $I_\alpha^{\downarrow\downarrow}(X:Y)_P$~\cite{burri2024doublyminimizedpetzrenyi}. 
Hence, the largest eigenvalue formulation mentioned in the previous paragraph also applies to the doubly minimized PRMI of order $\alpha=1/2$ if evaluated for a CC state. 
This naturally raises the question of whether the restriction to CC states can be omitted. 
We thus ask:

\begin{description}
    \item[Question~3] Does a largest eigenvalue formulation for $I_{1/2}^{\downarrow\downarrow}(A:B)_\rho$ exist that holds for all bipartite quantum states $\rho_{AB}$?
\end{description}

\textbf{Outline.} The remainder of this paper is structured as follows. 

Section~\ref{sec:preliminaries} contains some preliminaries. 
First, we explain our notation~(\ref{ssec:notation}). 
Then, we provide definitions and discuss properties related to 
entropies, divergences, and R\'enyi information measures~(\ref{ssec:entropies}), 
the computational basis, canonical purification, complex conjugation, and the swap operator~(\ref{ssec:swap}), 
and the R\'enyi reflected entropy~(\ref{ssec:renyi-reflected}). 

Section~\ref{sec:main} focuses on the presentation of our main result. 
The proof of our main result takes as its starting point the fact that $S_R^{(\infty)}(A:B)_\rho$ is, by definition, determined by the largest eigenvalue of $\hat{\rho}_{AA^*}$. 
Here, $\hat{\rho}_{AA^*}$ denotes the marginal state on 
$AA^*$ of the canonical purification of $\rho_{AB}$. 
As a first proof step, we then show that also $I_{1/2}^{\downarrow\downarrow}(A:B)_\rho$ can be expressed as a function of $\hat{\rho}_{AA^*}$~(\ref{ssec:rho-aa}). 
In order to derive this result, we show how the minimized generalized PRMI of order $1/2$ can be expressed in terms of $\hat{\rho}_{AA^*}$ (Proposition~\ref{prop:i12-omega}), which immediately implies a corresponding expression for the doubly minimized PRMI of order $1/2$ as a special case (Corollary~\ref{cor:h12-i12}). 
To relate the two expressions for $I_{1/2}^{\downarrow\downarrow}(A:B)_\rho$ and  $S_R^{(\infty)}(A:B)_\rho$, we make use of a special property of $\hat{\rho}_{AA^*}$ that we call \emph{$\mathcal{CF}$-invariance}. 
It means that $\hat{\rho}_{AA^*}$ is invariant under the joint application of complex conjugation ($\mathcal{C}$) and swapping ($\mathcal{F}$). 
In Appendix~\ref{app:invariance}, we develop basic techniques for handling $\mathcal{CF}$-invariant linear operators. 
By applying these techniques to our case, we obtain the desired equality in~\eqref{eq:main}, stated as Theorem~\ref{thm:srinf-i12} (\ref{ssec:srinf-i12}). 

Section~\ref{sec:implications} explores implications of our main result. 
This includes 
bounds on the min-reflected entropy~(\ref{ssec:bounds}), 
the correspondence between the maximality of the R\'enyi reflected entropy and maximal entanglement (\ref{ssec:max}), 
a comparison of the (R\'enyi) reflected entropy with the entanglement of purification (\ref{ssec:comparison_eop}), 
an operational interpretation of the min-reflected entropy (\ref{ssec:operational}), 
and a comparison of our result with an analogous result for the classical case in previous work~(\ref{ssec:classical}). 
In these parts, we also answer the three questions posed above. 

Section~\ref{sec:definition} addresses an afterthought concerning the definition of the R\'enyi reflected entropy, focusing on whether the use of the \emph{canonical} purification in the evaluation of the R\'enyi reflected entropy is essential. 
The key insight of this section is that the R\'enyi reflected entropy of order $n\in \mathbb{N}_{\geq 2}\cup \{\infty\}$ can be reformulated as a minimization over a larger class of purifications, where no explicit reference to the canonical purification remains (Theorem~\ref{thm:minimized}). 
While this section is formally independent of our main result, it is qualitatively connected to it, as our main result involves the min-reflected entropy -- i.e., the R\'enyi reflected entropy of order $n=\infty$ -- to which Theorem~\ref{thm:minimized} applies. 

Section~\ref{sec:conclusion} contains concluding remarks.

\textbf{Related work.} 
In the fields of condensed matter physics, integrable systems, and holographic conformal field theories, several works suggest that within the family of the R\'enyi mutual information of order $\alpha\in [0,\infty]$, the case $\alpha=1/2$ holds special significance, as it corresponds to twice the logarithmic negativity in various quantum systems under suitable conditions~\cite{alba2019quantum,kudlerflam2020correlation,kudlerflam2021quasi,bertini2022entanglement,
gruber2020time,ruggiero2022quantum,turkeshi2022enhanced} and is also related to the R\'enyi reflected entropy of order $1/2$~\cite{kudlerflam2019entanglement,kusuki2019derivation,kudlerflam2020correlation}. 
Importantly, all these works refer to the R\'enyi mutual information defined as $I_\alpha^S(A:B)_\rho\coloneqq H_\alpha(A)_\rho + H_\alpha(B)_\rho-H_\alpha(AB)_\rho$, which is simply the sum of three R\'enyi entropies. 
Owing to its straightforward definition, $I_\alpha^S(A:B)_\rho$ is often easier to compute than other types of R\'enyi mutual information. 
However, a significant downside is that it does not qualify as a correlation measure (except for $\alpha=1$), as it can increase under local operations~\cite{linden2013structure,kudlerflam2023renyi1,kudlerflam2023renyi}. 
In contrast, the doubly minimized PRMI of order $\alpha$ is non-increasing under local operations for all $\alpha\in [0,2]$ and exhibits several other natural properties~\cite{burri2024doublyminimizedpetzrenyi}. 
Our main result implies that, within the family of the doubly minimized PRMI of order $\alpha$, the case $\alpha=1/2$ stands out (when compared to R\'enyi reflected entropies), as it exactly equals the min-reflected entropy for general quantum states. 
Since there is no clear link between $I_{1/2}^{\downarrow\downarrow}(A:B)_\rho$ and $I_{1/2}^S(A:B)_\rho$, our work does not have an immediate connection to the aforementioned studies.

\section{Preliminaries}\label{sec:preliminaries}
\subsection{Notation}\label{ssec:notation}
The set of natural numbers strictly smaller than $n\in \mathbb{N}$ is denoted as $[n]\coloneqq \{0,1,\dots, n-1\}$. 
``$\log$'' denotes the natural logarithm. 

Throughout this paper, all Hilbert spaces are assumed to have finite dimension for simplicity. 
The dimension of a Hilbert space $A$ is denoted by $d_A\equiv \dim (A)$. 
The tensor product of two Hilbert spaces $A$ and $B$ is denoted by $A\otimes B$ or $AB$. 
The set of linear maps from $A$ to $B$ is denoted by $\mathcal{L}(A,B)$ and we define $\mathcal{L}(A)\coloneqq \mathcal{L}(A,A)$. 
To simplify the notation, identity operators are sometimes omitted, i.e., 
for any $X_A\in \mathcal{L}(A)$, ``$X_A$'' may also be interpreted as $X_A\otimes 1_B\in \mathcal{L}(AB)$. 
Furthermore, subscripts may be omitted for simplicity; 
for example, $\ket{\phi}\equiv \ket{\phi}_A$ for $\ket{\phi}_A\in A$, 
and $X\equiv X_A$ for $X_A\in \mathcal{L}(A)$. 

The kernel, image, and rank of $X\in \mathcal{L}(A,B)$ are denoted as $\ker(X),\im(X),$ and $\rank(X)$. 
The support of $X\in \mathcal{L}(A,B)$ is denoted as $\supp(X)$ and is defined as the orthogonal complement of the kernel of $X$. 
For $X,Y\in \mathcal{L}(A,B)$, $X\ll Y$ is true iff $\ker(Y)\subseteq \ker(X)$. 
For $X,Y\in \mathcal{L}(A,B)$, $X\not\perp Y$ is true iff $\supp(X)\cap \supp(Y)$ contains at least one non-zero vector. 
The Hermitian adjoint of $X\in \mathcal{L}(A,B)$ is denoted by $X^\dagger\in \mathcal{L}(B,A)$.

The spectrum of $X\in \mathcal{L}(A)$ is denoted as $\spec(X)$. 
The trace of $X\in \mathcal{L}(A)$ is denoted as $\tr[X]$, and the partial trace over $A$ as $\tr_A$. 
The commutator of $X,Y\in \mathcal{L}(A)$ is denoted as $[X,Y]\coloneqq XY-YX$. 
For $X\in \mathcal{L}(A)$, $X\geq 0$ is true iff $X$ is positive semidefinite.
We employ the generalized inverse to take powers of positive semidefinite operators that do not have full support. 
Thus, for any positive semidefinite $X\in \mathcal{L}(A)$, $X^p$ is defined for $p\in \mathbb{R}$ by taking the power on the support of $X$. 
For $p=1/2$, the square root symbol is sometimes used ($\sqrt{X}\coloneqq X^{1/2}$). 
The operator absolute value of $X\in \mathcal{L}(A,B)$ is denoted as $\lvert X\rvert\coloneqq (X^\dagger X)^{1/2}$. 
The Schatten $p$-norm of $X\in \mathcal{L}(A,B)$ is denoted as $\|X \|_p\coloneqq (\tr[\lvert X\rvert^p])^{1/p}$ for $p\in [1,\infty)$, 
and as $\lVert X\rVert_{\infty}\coloneqq \sqrt{\max (\spec(X^\dagger X))}$ for $p=\infty$. 
The Schatten $p$-quasi-norm is denoted as $\lVert X\rVert_p\coloneqq (\tr[|X|^p])^{1/p}$ for $p\in (0,1)$.

The set of unitary operators on $A$ is $\mathcal{U}(A)\coloneqq \{U\in \mathcal{L}(A):U^\dagger U=1_A\}$. 
The set of orthogonal projections on $A$ is denoted by $\Proj(A)\coloneqq\{P\in \mathcal{L}(A):P^\dagger = P = P^2\}$, and the set of orthogonal projections of rank $r\in \mathbb{N}$ by $\Proj_r(A)\coloneqq \{P\in \Proj(A):\mathrm{rank}(P)=r\}$.

The set of (quantum) states on $A$ is $\mathcal{S}(A)\coloneqq\{\rho\in \mathcal{L}(A):\rho\geq 0,\tr[\rho]=1\}$. 
The set of completely positive, trace-preserving linear maps from $\mathcal{L}(A)$ to $\mathcal{L}(B)$ is denoted by $\CPTP(A, B)$. 
Elements of this set are called (quantum) channels.

\subsection{Entropies, divergences, and R\'enyi information measures}\label{ssec:entropies}
The \emph{von Neumann entropy} of $\rho\in \mathcal{S}(A)$ is $H(A)_\rho \coloneqq -\tr[\rho\log \rho]$. 
The \emph{R\'enyi entropy of order $\alpha$} of $\rho\in \mathcal{S}(A)$ is defined as $H_\alpha (A)_\rho \coloneqq\frac{1}{1-\alpha}\log \tr[\rho^\alpha]$ for $\alpha \in (0,1)\cup (1,\infty)$, 
and for $\alpha \in \{0,1,\infty\}$ as the corresponding limits. 
The limit $\alpha\rightarrow 1$ coincides with the von Neumann entropy, i.e., $H_1(A)_\rho =H(A)_\rho$. 
The limit $\alpha\rightarrow \infty$ is given by
\begin{equation}\label{eq:h-inf}
H_\infty(A)_\rho 
=-\log \lVert \rho\rVert_\infty
=-\log \max ( \spec(\rho))
=-\log \max_{\substack{\ket{\phi}\in A:\\ \brak{\phi}=1}}\tr[\rho\proj{\phi}]
\end{equation}
and is called the \emph{min-entropy}. 

The \emph{mutual information} of $\rho_{AB}\in \mathcal{S}(AB)$ is $I(A:B)_\rho \coloneqq H(A)_\rho+H(B)_\rho-H(AB)_\rho$, 
and the \emph{conditional entropy} is $H(A|B)_\rho \coloneqq H(AB)_\rho -H(B)_\rho$. 
The \emph{entanglement of purification}~\cite{terhal2002entanglement} is defined for $\rho_{AB}\in \mathcal{S}(AB)$ as 
\begin{align}
E_P(A:B)_\rho \coloneqq \min_{\substack{\sigma_{ABC}\in \mathcal{S}(ABC):\\\tr_C[\sigma_{ABC}]=\rho_{AB}}}H(AC)_\sigma 
\end{align}
where $C$ denotes an arbitrary finite-dimensional Hilbert space, 
and it is understood that the minimization also ranges over the dimension of $C$. 
This measure is symmetric, $E_P(A:B)_\rho=E_P(B:A)_\rho$, and it is bounded as
\begin{equation}\label{eq:eop-bounds}
\frac{1}{2}I(A:B)_\rho\leq 
E_P(A:B)_\rho
\leq \min \{H(A)_\rho,H(B)_\rho\}
\leq \min \{H_0(A)_\rho,H_0(B)_\rho\}
\leq \log \min \{d_A,d_B\}.
\end{equation}

The \emph{quantum relative entropy} of $\rho\in \mathcal{S}(A)$ relative to any positive semidefinite $\sigma\in \mathcal{L}(A)$ is
\begin{equation}
D(\rho\| \sigma)\coloneqq \tr[\rho (\log\rho -\log\sigma)]
\end{equation}
if $\rho\ll \sigma$, and $D(\rho\| \sigma)\coloneqq\infty$ else. 
Using the quantum relative entropy, the mutual information and the conditional entropy of $\rho_{AB}\in \mathcal{S}(AB)$ can be expressed as follows~\cite{mueller2013quantum,tomamichel2014relating,gupta2014multiplicativity,hayashi2016correlation}.
\begin{align}
I(A:B)_\rho 
&= D(\rho_{AB}\| \rho_A\otimes \rho_B)
=\inf_{\tau_B\in \mathcal{S}(B)}D(\rho_{AB}\| \rho_A\otimes \tau_B)
=\inf_{\substack{\sigma_A\in \mathcal{S}(A),\\ \tau_B\in \mathcal{S}(B)}}D(\rho_{AB}\| \sigma_A\otimes \tau_B)
\label{eq:i-divergence}
\\
H(A|B)_\rho
&= -D(\rho_{AB}\| 1_A\otimes \rho_B) 
=\sup_{\tau_B\in \mathcal{S}(B)} -D(\rho_{AB}\| 1_A\otimes \tau_B)
\label{eq:h-divergence}
\end{align}

The \emph{Petz (quantum R\'enyi) divergence of order $\alpha$} is defined for $\alpha\in [0,1)\cup (1,\infty)$, $\rho\in \mathcal{S}(A)$, and any positive semidefinite $\sigma\in \mathcal{L}(A)$ as
\begin{equation}
D_\alpha (\rho\| \sigma)\coloneqq \frac{1}{\alpha -1} \log \tr [\rho^\alpha \sigma^{1-\alpha}]
\end{equation}
if $(\alpha <1\land \rho\not\perp\sigma)\lor \rho\ll\sigma$, and $D_\alpha (\rho\| \sigma)\coloneqq\infty$ else.
$D_1$ is defined as the limit of $D_\alpha$ for $\alpha\rightarrow 1$. 
This limit coincides with the quantum relative entropy, i.e., 
$D_1 (\rho\| \sigma)=D(\rho\|\sigma )$~\cite{tomamichel2016quantum}.

The \emph{minimized generalized Petz R\'enyi mutual information (PRMI) of order $\alpha$} 
of $\rho_{AB}\in \mathcal{S}(AB)$ relative to any positive semidefinite $\sigma_A\in \mathcal{L}(A)$ is defined for $\alpha\in [0,\infty)$ as~\cite{hayashi2016correlation}
\begin{align}\label{eq:mi-generalized}
I_{\alpha}^{\downarrow}(\rho_{AB}\| \sigma_A)
\coloneqq \inf_{\tau_B\in \mathcal{S}(B)}D_{\alpha}(\rho_{AB}\| \sigma_A\otimes \tau_B) .
\end{align}
If $(\alpha\in (0,1)\land \rho_A\not\perp\sigma_A)\lor (\alpha\in[1,\infty)\land \rho_A\ll \sigma_A)$, then the optimizer is uniquely given by 
$\tau_B=(\tr_A[\rho_{AB}^\alpha\sigma_A^{1-\alpha}])^{\frac{1}{\alpha}}/\tr[(\tr_A[\rho_{AB}^\alpha\sigma_A^{1-\alpha}])^{\frac{1}{\alpha}}]$~\cite{hayashi2016correlation}.
By inserting this optimizer in~\eqref{eq:mi-generalized}, it follows that for any $\alpha\in (0,1)\cup (1,\infty)$
\begin{equation}\label{eq:i-gen-explicit}
I_{\alpha}^{\downarrow}(\rho_{AB}\| \sigma_A)
=\frac{\alpha}{\alpha -1}\log \tr [(\tr_A[\rho_{AB}^\alpha \sigma_A^{1-\alpha}] )^{\frac{1}{\alpha}}]
\end{equation}
if $(\alpha \in (0,1)\land\rho_A\not\perp \sigma_A)\lor \rho_A\ll \sigma_A$, and $I_{\alpha}^{\downarrow}(\rho_{AB}\| \sigma_A)=\infty$ else.

To obtain a R\'enyi version of the mutual information, any of the three expressions in~\eqref{eq:i-divergence} can serve as a starting point, yielding different R\'enyi generalizations of the mutual information when combined with a suitable R\'enyi divergence~\cite{gupta2014multiplicativity,hayashi2016correlation,burri2024doublyminimizedpetzrenyi,burri2024doublyminimizedsandwichedrenyi}. 
Of these, only the following two types of R\'enyi mutual information will be of interest in this work. 
In analogy to the second-to-last expression in~\eqref{eq:i-divergence}, the \emph{singly minimized PRMI of order $\alpha$}~\cite{hayashi2016correlation} is defined for $\alpha \in [0,\infty)$ as
\begin{equation}\label{eq:prmi1}
I_\alpha^{\uparrow\downarrow} (A:B)_\rho 
\coloneqq \inf_{\tau_B\in \mathcal{S}(B)}D_\alpha (\rho_{AB}\| \rho_A\otimes \tau_B ) 
=I_\alpha^{\downarrow}(\rho_{AB}\| \rho_A).
\end{equation}
In analogy to the last expression in~\eqref{eq:i-divergence}, the \emph{doubly minimized PRMI of order $\alpha$} is defined for $\alpha \in [0,\infty)$ as~\cite{burri2024doublyminimizedpetzrenyi}
\begin{equation}\label{eq:prmi2}
I^{\downarrow\downarrow}_\alpha (A:B)_\rho \coloneqq \inf_{\substack{\sigma_A\in \mathcal{S}(A), \\ \tau_B\in \mathcal{S}(B) }} D_\alpha (\rho_{AB}\| \sigma_A\otimes \tau_B ) 
=\inf_{\sigma_A\in \mathcal{S}(A)} I_\alpha^{\downarrow}(\rho_{AB}\| \sigma_A).
\end{equation}

To obtain a R\'enyi version of the conditional entropy, any of the two expressions in~\eqref{eq:h-divergence} can serve as a starting point, yielding different R\'enyi generalizations of the conditional entropy when combined with a suitable R\'enyi divergence~\cite{tomamichel2014relating}. 
Of these, only the following type of R\'enyi conditional entropy will be of interest in this work. 
In analogy to the last expression in~\eqref{eq:h-divergence}, the \emph{maximized Petz R\'enyi conditional entropy of order $\alpha$} is defined for $\alpha \in [0,\infty)$ as
\begin{equation}\label{eq:h-alpha}
H_\alpha^{\uparrow} (A|B)_\rho
\coloneqq \sup_{\tau_B\in \mathcal{S}(B)}- D_\alpha (\rho_{AB}\| 1_A\otimes \tau_B)
=-I_{\alpha}^{\downarrow}(\rho_{AB}\| 1_A)
=\log d_A-I_{\alpha}^{\downarrow}(\rho_{AB}\| 1_A/d_A).
\end{equation}

\subsection{Computational basis, canonical purification, complex conjugation, and swap operator}\label{ssec:swap}
We assume that each Hilbert space $A$ is equipped with an arbitrary but fixed orthonormal basis $\{\ket{i}_A\}_{i\in [d_A]}$, the \emph{computational basis of $A$}. 
The computational basis of any tensor product $A\otimes B$ is assumed to coincide with the tensor product of the computational bases of $A$ and $B$. 
If $A$ is a Hilbert space, then $A^*$ denotes a Hilbert space isomorphic to $A$. 
Using the computational basis, 
we define the \emph{unnormalized canonical maximally entangled state on $AA^*$} as 
$\ket{\Omega}_{AA^*}\coloneqq \sum_{j\in [d_A]}\ket{j}_A\otimes \ket{j}_{A^*}$. 
Since the computational basis of a tensor product of Hilbert spaces coincides with the tensor product of the respective computational bases, we have for any $n\in \mathbb{N}_{>0}$
\begin{equation}
\ket{\Omega}_{A_1\dots A_nA_1^*\dots A_n^*}
=\bigotimes\limits_{j=1}^n \ket{\Omega}_{A_jA_j^*} .
\end{equation}

Hilbert spaces of the same dimension are identified through their respective computational bases. 
Specifically, the isomorphism $A\simeq A^*$ is determined by 
$U\coloneqq \sum_{j\in [d_A]}\ket{j}_{A^*}\bra{j}_A\in \mathcal{L}(A,A^*)$, 
where we adopt the following notation conventions.
\begin{itemize}
\item For any $\ket{\phi}_A\in A$, ``$\ket{\phi}_{A^*}$'' is understood as $\ket{\phi}_{A^*}\coloneqq U\ket{\phi}_A\in A^*$. 
\item For any $\bra{\phi}_A\in \mathcal{L}(A,\mathbb{C})$, ``$\bra{\phi}_{A^*}$'' is understood as $\bra{\phi}_{A^*}\coloneqq \bra{\phi}_AU^\dagger \in \mathcal{L}(A^*,\mathbb{C})$. 
\item For any $X_A\in \mathcal{L}(A)$, ``$X_{A^*}$'' is understood as $X_{A^*}\coloneqq UX_AU^{\dagger}\in \mathcal{L}(A^*)$.
\end{itemize}

The \emph{canonical purification of $\rho_A\in \mathcal{S}(A)$} is defined as
\begin{equation}\label{eq:def-cp}
\ket{\hat{\rho}}_{AA^*} \coloneqq \sqrt{\rho_A}\ket{\Omega}_{AA^*}\in AA^* .
\end{equation}
Note that this definition is also applicable to multipartite states. 
For instance,~\eqref{eq:def-cp} implies that the canonical purification of a bipartite quantum state $\rho_{AB}\in \mathcal{S}(AB)$ is given by
\begin{equation}\label{eq:purification-abab}
\ket{\hat{\rho}}_{ABA^*B^* } 
=\sqrt{\rho_{AB}}\ket{\Omega}_{AA^*}\otimes \ket{\Omega}_{BB^*} .
\end{equation}

The \emph{complex conjugate of} $X_{A}\in \mathcal{L}(A)$ is denoted by $X_A^*$ and is taken with respect to the computational basis of $A$, 
i.e., $X_A^*\coloneqq\sum_{j,k\in [d_A]} \bra{j}_AX_A\ket{k}_A^*\ketbra{j}{k}_A$. 
Similarly, the \emph{complex conjugate of} $\ket{\phi}_A\in A$ is given by 
$\ket{\phi}_A^*\coloneqq\sum_{j\in [d_A]}\braket{j}{\phi}_A^*\ket{j}_A$.

The \emph{swap operator (or: flip operator)} is defined as 
$\swapf_{AA^*} \coloneqq\sum_{j,k\in [d_A]}\ketbra{j}{k}_{A}\otimes \ketbra{k}{j}_{A^*}\in \mathcal{L}(AA^*)$.
This implies that $\swapf_{AA^*} \ket{\phi}_A\otimes \ket{\psi}_{A^*}=\ket{\psi}_A\otimes \ket{\phi}_{A^*}$
for all $\ket{\phi}_{A}\in A,\ket{\psi}_{A^*}\in A^*$.
The swap operator is self-adjoint, unitary, and invariant under complex conjugation $(\swapf_{AA^*} ^*=\swapf_{AA^*} )$. 
The action of the swap operator on $X_{AA^*}\in \mathcal{L}(AA^*)$ is denoted as
$X_{AA^*}^{\swapf_{AA^*} }\coloneqq \swapf_{AA^*}  X_{AA^*} \swapf_{AA^*} \in \mathcal{L}(AA^*)$.

\subsection{R\'enyi reflected entropy}\label{ssec:renyi-reflected}

The \emph{R\'enyi reflected entropy of order $n\in [0,\infty]$ for $m\in [0,\infty)$} is defined for $\rho_{AB}\in \mathcal{S}(AB)$ as~\cite{dutta2019canonical,akers2022page}
\begin{equation}\label{eq:def-sr}
S_R^{(m,n)}(A:B)_\rho \coloneqq H_n (AA^*)_{\hat{\rho}^{(m)}}
\end{equation}
where $\hat{\rho}_{AA^*}^{(m)}\coloneqq \tr_{BB^*}[\proj{\hat{\rho}^{(m)}}_{ABA^*B^*}]$ and 
\begin{equation}\label{eq:rho-aa}
\ket{\hat{\rho}^{(m)}}_{ABA^*B^*}\coloneqq 
\frac{1}{\sqrt{\tr[\rho_{AB}^m]}}\rho_{AB}^{m/2} \ket{\Omega}_{AA^*}\otimes\ket{\Omega}_{BB^*}  .
\end{equation}
Note that $\ket{\hat{\rho}^{(m)}}_{ABA^*B^*}$ is the canonical purification of $\rho_{AB}^m/\tr[\rho_{AB}^m]$. 

The \emph{R\'enyi reflected entropy of order $n\in [0,\infty]$} is $S_R^{(n)}(A:B)_\rho\coloneqq S_R^{(1,n)}(A:B)_\rho$~\cite{hayden2023reflected}. 
The \emph{reflected entropy}~\cite{dutta2019canonical} is $S_R (A:B)_\rho \coloneqq S_R^{(1)}(A:B)_\rho$, 
and we refer to $S_R^{(\infty)}(A:B)_\rho$ as the \emph{min-reflected entropy}. 

The $(m,n)$-R\'enyi reflected entropy of $\rho_{AB}$ can always be expressed as the $n$-R\'enyi reflected entropy of another state that depends on $m$. 
To be precise, we have for any $m\in [0,\infty),n\in [0,\infty],\rho_{AB}\in \mathcal{S}(AB)$
\begin{align}\label{eq:sr-mn-1n}
S_R^{(m,n)}(A:B)_\rho
=S_R^{(n)}(A:B)_{\sigma^{(m)}}
\qquad\text{where}\qquad
\sigma_{AB}^{(m)}\coloneqq \rho_{AB}^m/\tr[\rho_{AB}^m].
\end{align}
By means of~\eqref{eq:sr-mn-1n}, several results regarding the $n$-R\'enyi reflected entropy can be readily generalized to statements about the $(m,n)$-R\'enyi reflected entropy. 
For this reason, the remainder of this section focuses on the $n$-R\'enyi reflected entropy.

Below, we enumerate several properties of the R\'enyi reflected entropy, most of which are well-known~\cite{dutta2019canonical}. 
Since these properties directly follow from the definition of the R\'enyi reflected entropy (primarily due to corresponding properties of R\'enyi entropies), they are stated without proof. 

Let $\rho_{AB}\in \mathcal{S}(AB)$. 
Then all of the following hold.

\begin{enumerate}[label=(\alph*)]
\item \emph{Monotonicity:} For all $\alpha,\beta\in [0,\infty]$ such that $\alpha \leq \beta$ holds 
$S_R^{(\alpha)}(A:B)_\rho\geq S_R^{(\beta)}(A:B)_\rho$.
\item \emph{Continuity, Non-negativity:} The function $[0,\infty)\rightarrow [0,\infty),\alpha\mapsto S_R^{(\alpha)}(A:B)_\rho$ is continuous, and 
$S_R^{(\infty)}(A:B)_\rho=\lim_{\alpha\rightarrow\infty}S_R^{(\alpha)}(A:B)_\rho$.
\item \emph{Symmetry:} $S_R^{(\alpha)}(A:B)_\rho = S_R^{(\alpha)}(B:A)_\rho$ for all $\alpha\in [0,\infty]$.
\item \emph{Bounds:} For all $\alpha\in [0,\infty]$ 
\begin{align}\label{eq:sr-bounds}
S_R^{(\alpha)}(A:B)_\rho 
\leq S_R^{(0)}(A:B)_\rho
\leq 2\log\min\{\rank(\rho_A),\rank(\rho_B)\}
\leq 2\log \min\{d_A,d_B\} .
\end{align}
\item \emph{Invariance under local isometries:} 
Let $V\in \mathcal{L}(A,C),W\in \mathcal{L}(B,D)$ be isometries. 
Then 
\begin{align}
S_R^{(\alpha)}(C:D)_{V\otimes W\rho_{AB} V^\dagger \otimes W^\dagger} =S_R^{(\alpha)}(A:B)_\rho
\qquad\forall\alpha\in [0,\infty].
\end{align}
\item \emph{Additivity over tensor product:} 
Let $\sigma_{CD}\in \mathcal{S}(CD)$. Then 
\begin{equation}
S_R^{(\alpha)}(AC:BD)_{\rho\otimes\sigma}
= S_R^{(\alpha)}(A:B)_{\rho} + S_R^{(\alpha)}(C:D)_{\sigma} 
\qquad\forall\alpha\in [0,\infty].
\end{equation}
\item \emph{Reflected entropy:} 
Let $\ket{\hat{\rho}}_{ABA^*B^*}$ be the canonical purification of $\rho_{AB}$ and let $\hat{\rho}_{ABA^*}\coloneqq \tr_{B^*}[\proj{\hat{\rho}}_{ABA^*B^*}]$. 
Then, $S_R(A:B)_\rho=I(AA^*:B)_{\hat{\rho}}$  
and 
\begin{equation}\label{eq:i-leq-sr}
0\leq I(A:B)_\rho \leq S_R(A:B)_\rho \leq 2\min \{H(A)_\rho,H(B)_\rho\} .
\end{equation}
\item \emph{Product states:} 
For any $\alpha\in [0,\infty]:$ $\rho_{AB}=\rho_A\otimes \rho_B$ iff $S_R^{(\alpha)}(A:B)_\rho = 0$.
\item \emph{Pure states:} If there exists $\ket{\rho}_{AB}\in AB$ such that $\rho_{AB}=\proj{\rho}_{AB}$, then $S_R^{(\alpha)}(A:B)_\rho=2H_\alpha (A)_\rho$ for all $\alpha\in [0,\infty]$.
\item \emph{CC states:} \sloppy 
Let $P_{XY}$ be the joint PMF of two random variables $X,Y$ over $[d_A],[d_B]$. 
Let $M_{x,y}\coloneqq \sqrt{P_{XY}(x,y)}$ for all $x\in [d_A],y\in [d_B]$ 
and let $M\coloneqq (M_{x,y})_{x\in [d_A],y\in [d_B]}$ be a $(d_A\times d_B)$-matrix. 
Let $M^T$ denote the transpose of $M$. 
If there exist orthonormal bases $\{\ket{a_x}_A\}_{x\in [d_A]},\{\ket{b_y}_B\}_{y\in [d_B]}$ for $A,B$ such that 
$\rho_{AB}=\sum_{x\in [d_A],y\in [d_B]} P_{XY}(x,y)\proj{a_x,b_y}_{AB}$, then 
\begin{align}
S_R^{(\alpha)}(A:B)_\rho
&=\frac{1}{1-\alpha}\log \lVert MM^T \rVert_\alpha^\alpha
=\frac{1}{1-\alpha}\log \lVert M \rVert_{2\alpha}^{2\alpha} 
\qquad \forall \alpha\in (0,1)\cup (1,\infty),
\\
S_R^{(\infty)}(A:B)_\rho
&=-\log \lVert MM^T \rVert_\infty
=-2\log \lVert M\rVert_{\infty}.
\label{eq:reflected-cc}
\end{align}
\item \emph{Copy-CC states:} 
Let $P_{X}$ be the PMF of a random variable $X$ over $[\min\{d_A,d_B\}]$. 
If there exist orthonormal bases $\{\ket{a_x}_A\}_{x\in [d_A]},\{\ket{b_y}_B\}_{y\in [d_B]}$ for $A,B$ such that 
$\rho_{AB}= \sum_{x\in [\min\{d_A,d_B\}]} P_{X}(x)\proj{a_x,b_x}_{AB}$, then 
$S_R^{(\alpha)}(A:B)_\rho=H_\alpha (A)_\rho$ for all $\alpha\in [0,\infty]$.
\end{enumerate}

\section{Main result}\label{sec:main}

\subsection{Minimized generalized PRMI of order 1/2 from canonical purification}\label{ssec:rho-aa}
The following proposition shows how the minimized generalized PRMI of order 1/2 can be expressed in terms of $\hat{\rho}_{AA^*}$ instead of $\rho_{AB}$. 
The proof of Proposition~\ref{prop:i12-omega} is given in Appendix~\ref{app:proof-omega}.
The evaluation of Proposition~\ref{prop:i12-omega} for different states $\sigma_A$ yields the subsequent corollary. 
Of these results, only~\eqref{eq:i122-rhoaa} in Corollary~\ref{cor:h12-i12} will be used in the rest of the paper. 
The other results are presented for completeness. 

\begin{prop}[Minimized generalized PRMI of order 1/2 in terms of $\hat{\rho}_{AA^*}$]\label{prop:i12-omega}
Let $\rho_{AB}\in \mathcal{S}(AB)$, 
let $\ket{\hat{\rho}}_{ABA^*B^*}$ be its canonical purification, 
and let $\hat{\rho}_{AA^*}\coloneqq \tr_{BB^*}[\proj{\hat{\rho}}_{ABA^*B^*}]$. 
Let $\sigma_A\in \mathcal{L}(A)$ be positive semidefinite. Then
\begin{equation}
I_{1/2}^{\downarrow}(\rho_{AB}\| \sigma_A )
=-\log \bra{\Omega}_{AA^*} \sqrt{\sigma_A}\hat{\rho}_{AA^*}\sqrt{\sigma_A}\ket{\Omega}_{AA^*}
\end{equation}
if $\rho_A\not\perp \sigma_A$, and $I_{1/2}^{\downarrow}(\rho_{AB}\| \sigma_A )=\infty$ else.
\end{prop}

\begin{cor}[Petz R\'enyi information measures of order 1/2 in terms of $\hat{\rho}_{AA^*}$]\label{cor:h12-i12}
Let $\rho_{AB}\in \mathcal{S}(AB)$, 
let $\ket{\hat{\rho}}_{ABA^*B^*}$ be its canonical purification, 
and let $\hat{\rho}_{AA^*}\coloneqq \tr_{BB^*}[\proj{\hat{\rho}}_{ABA^*B^*}]$. 
Then the following identities hold.
\begin{align}
\log d_A- H_{1/2}^{\uparrow}(A|B)_\rho 
= I_{1/2}^{\downarrow}(\rho_{AB}\| 1_A/d_A) 
&=-\log \bra{\Omega}_{AA^*}\frac{1}{\sqrt{d_A}}\hat{\rho}_{AA^*}\frac{1}{\sqrt{d_A}}\ket{\Omega}_{AA^*}
\label{eq:h12-rhoaa}\\
I_{1/2}^{\uparrow\downarrow}(A:B)_\rho 
= I_{1/2}^{\downarrow}(\rho_{AB}\| \rho_A )
&= - \log \bra{\Omega}_{AA^*}\sqrt{\rho_A}\hat{\rho}_{AA^*}\sqrt{\rho_A}\ket{\Omega}_{AA^*}
\label{eq:i121-rhoaa}\\
I_{1/2}^{\downarrow\downarrow}(A:B)_\rho
=\min_{\sigma_A\in \mathcal{S}(A)} I_{1/2}^{\downarrow}(\rho_{AB}\| \sigma_A ) 
&= -\log \max_{\sigma_A\in \mathcal{S}(A)} \bra{\Omega}_{AA^*}\sqrt{\sigma_A}\hat{\rho}_{AA^*}\sqrt{\sigma_A}\ket{\Omega}_{AA^*}
\label{eq:i122-rhoaa}
\end{align}
\end{cor}

\begin{rem}[Information content in $\hat{\rho}_{AA^*}$]
Corollary~\ref{cor:h12-i12} shows how three different Petz R\'enyi information measures of order 1/2 can be expressed as a function of $\hat{\rho}_{AA^*}$.
In particular, all three quantities can be expressed as $-\log \bra{e_0}\hat{\rho}_{AA^*}\ket{e_0}$ for a suitable choice of the unit vector $\ket{e_0}\in AA^*:$
\begin{itemize}
\item 
By~\eqref{eq:h12-rhoaa}, $\log d_A-H_{1/2}^{\uparrow}(A|B)_\rho=-\log \bra{e_0}\hat{\rho}_{AA^*}\ket{e_0}$ 
if $\ket{e_0}_{AA^*}$ is defined as the canonical purification of the maximally mixed state $1_A/d_A$, 
i.e., $\ket{e_0}_{AA^*}\coloneqq \sqrt{1_A/d_A}\ket{\Omega}_{AA^*}$.
\item 
By~\eqref{eq:i121-rhoaa}, $I_{1/2}^{\uparrow\downarrow}(A:B)_\rho=-\log \bra{e_0}\hat{\rho}_{AA^*}\ket{e_0}$ 
if $\ket{e_0}_{AA^*}$ is defined as the canonical purification of $\rho_A$, 
i.e., $\ket{e_0}_{AA^*}\coloneqq \sqrt{\rho_A}\ket{\Omega}_{AA^*}$.
\item 
By~\eqref{eq:i122-rhoaa}, $I_{1/2}^{\downarrow\downarrow}(A:B)_\rho=-\log \bra{e_0}\hat{\rho}_{AA^*}\ket{e_0}$ 
if $\ket{e_0}_{AA^*}$ is defined as the canonical purification of an optimizer $\sigma_A$, 
i.e., $\ket{e_0}_{AA^*}\coloneqq \sqrt{\sigma_A}\ket{\Omega}_{AA^*}$.
\end{itemize}
Consequently, each of these three quantities is determined by a single diagonal matrix element of $\hat{\rho}_{AA^*}$ --- if represented in a suitable orthonormal basis for $AA^*$. 
More explicitly, if $\hat{\rho}_{AA^*}$ is represented in an orthonormal basis whose $0$'th element is the respective $\ket{e_0}_{AA^*}$, then the $(0,0)$ entry of the matrix representation of $\hat{\rho}_{AA^*}$ determines~\eqref{eq:h12-rhoaa},~\eqref{eq:i121-rhoaa}, and~\eqref{eq:i122-rhoaa}, respectively. 
Remarkably, the choice of the orthonormal basis has a significant impact on the type of information measure that arises: 
For the first choice of $\ket{e_0}_{AA^*}$, the resulting quantity corresponds to a \emph{conditional entropy} of $A$ given $B$, whereas for the second and third choices, the resulting quantity is a \emph{mutual information} between $A$ and $B$. 
Clearly, the conditional entropy $H_{1/2}^\uparrow(A|B)_\rho$ is not monotonic under the inclusion of subsystems in the first argument. 
Thus, Corollary~\ref{cor:h12-i12} implies that $\hat{\rho}_{AA^*}$ contains information about two qualitatively distinct types of information measures, only one of which is a correlation measure.  
This observation may help explain why certain functions of $\hat{\rho}_{AA^*}$, such as the R\'enyi reflected entropy $S_R^{(n)}(A:B)_\rho$ of order $n\in (0,2)$~\cite{hayden2023reflected}, fail to be correlation measures. 
Whether this reasoning can be formalized into a precise quantitative argument remains an open question for future research.
\end{rem}

\subsection{Min-reflected entropy = doubly minimized Petz R\'enyi mutual information of order 1/2}\label{ssec:srinf-i12}

\begin{thm}\label{thm:srinf-i12}
Let $\rho_{AB}\in \mathcal{S}(AB)$. Then $S_R^{(\infty)}(A:B)_\rho= I_{1/2}^{\downarrow\downarrow}(A:B)_{\rho}$.
\end{thm}

In order to prove this theorem, we will employ techniques for $\mathcal{CF}$-invariant objects, which are previously established in Appendix~\ref{app:invariance}. 
The proof of Theorem~\ref{thm:srinf-i12} is given in Appendix~\ref{app:proof-thm}, and we comment on the role of $\mathcal{CF}$-invariance in this proof in Appendix~\ref{app:thm-remarks}. 
For completeness, we provide results on how optimizers of the defining expression for $S_R^{(\infty)}(A:B)_\rho$ are related to optimizers of the defining expression for $I_{1/2}^{\downarrow\downarrow}(A:B)_\rho$ in Appendix~\ref{app:thm-relations}.

Since the $(m,\infty)$-R\'enyi reflected entropy is generally related to the $\infty$-R\'enyi reflected entropy of another state that depends on $m$, see~\eqref{eq:sr-mn-1n}, we obtain the following corollary of Theorem~\ref{thm:srinf-i12}. 
\begin{cor}[R\'enyi reflected entropies of order $n=\infty$]
Let $\rho_{AB}\in \mathcal{S}(AB)$.
Let $m\in [0,\infty)$ and let $\sigma^{(m)}_{AB}\coloneqq \rho_{AB}^{m}/\tr[\rho_{AB}^{m}]$. 
Then
$S_R^{(m,\infty)}(A:B)_\rho 
=S_R^{(\infty)}(A:B)_{\sigma^{(m)}}
=I_{1/2}^{\downarrow\downarrow}(A:B)_{\sigma^{(m)}}$.
\end{cor}

\section{Implications of main result}\label{sec:implications}

\subsection{Bounds on min-reflected entropy}\label{ssec:bounds}

Since various properties of $I_{1/2}^{\downarrow\downarrow}(A:B)_\rho$ are known~\cite{burri2024doublyminimizedpetzrenyi}, 
an obvious implication of our main result is that these properties carry over to the min-reflected entropy. 
In particular, this yields new bounds on the min-reflected entropy, as expressed in the following corollary of Theorem~\ref{thm:srinf-i12}. 
The proof of this corollary is given in Appendix~\ref{app:proof_cor_bounds}. 
\begin{cor}[Lower and upper bounds on min-reflected entropy]\label{cor:bounds}
Let $\rho_{AB}\in \mathcal{S}(AB)$. Then
\begin{align}
\frac{1}{2}S_R^{(2)}(A:B)_\rho 
\leq S_R^{(\infty)}(A:B)_\rho
=I_{1/2}^{\downarrow\downarrow}(A:B)_\rho
&\leq \min\{I_{1/2}^{\uparrow\downarrow}(A:B)_\rho ,\log d_A- H_{1/2}^\uparrow (A|B)_\rho \}
\label{eq:bounds-1}\\
&\leq I_{1/2}^{\uparrow\downarrow}(A:B)_\rho 
\leq I(A:B)_\rho\leq S_R(A:B)_\rho .
\label{eq:bounds-2}
\end{align}
\end{cor}

\subsection{Maximality of R\'enyi reflected entropy corresponds to maximal entanglement}\label{ssec:max}
The following corollary is a direct consequence of Corollary~\ref{cor:bounds}, as shown in Appendix~\ref{app:proof_cor_equivalence}.
\begin{cor}[Maximality of R\'enyi reflected entropy]\label{cor:equivalence}
Let $\rho_{AB}\in \mathcal{S}(AB),r_A\coloneqq \rank(\rho_A),r_B\coloneqq\rank(\rho_B)$.
Then, $S_R^{(\alpha)}(A:B)_\rho\leq 2\log r_A$ for all $\alpha\in [0,\infty]$.
Moreover, for any $\alpha\in (0,\infty]$, all of the following hold.
\begin{enumerate}[label=(\alph*)]
\item $S_R^{(\alpha)}(A:B)_\rho=2\log r_A$ iff
$I(A:B)_\rho=2\log r_A$.
\item $S_R^{(\alpha)}(A:B)_\rho=2\log \min\{r_A,r_B\}$ iff
$I(A:B)_\rho=2\log \min\{r_A,r_B\}$.
\item $S_R^{(\alpha)}(A:B)_\rho=2\log \min\{d_A,d_B\}$ iff
$I(A:B)_\rho=2\log \min\{d_A,d_B\}$.
\end{enumerate} 
\end{cor}

Previously, only arguments for the direction ``$\Leftarrow$'' of the equivalence in~(c) have been presented in the literature. 
For example, for $\alpha=1$, the direction ``$\Leftarrow$'' holds trivially due to the bound $I(A:B)_\rho \leq S_R(A:B)\leq 2\log \min\{d_A,d_B\}$~\cite{dutta2019canonical}.
In addition, it is straightforward to show that if $\rho_{AB}$ is a pure, maximally entangled state, then the R\'enyi reflected entropy is maximal~\cite{dutta2019canonical}, which proves ``$\Leftarrow$'' for a special case. 
General arguments for the direction ``$\Rightarrow$'' appear to be absent in the existing literature. 
The equivalence in~(c) fills this gap, 
providing a rigorous justification for interpreting the maximality of the R\'enyi reflected entropy as maximal entanglement between $A$ and $B$.

\begin{rem}[Maximality of mutual information corresponds to maximal entanglement] \label{rem:max-mi}
States $\rho_{AB}\in \mathcal{S}(AB)$ that achieve the maximal value of the mutual information have a well-known operational meaning. 
Suppose $d_A\leq d_B$. Then, 
$I(A:B)_\rho = 2\log d_A$ iff there exist $\mathcal{D}_{B\rightarrow A^*}\in\CPTP(B,A^*)$ and $U_A\in \mathcal{U}(A)$ such that $\mathcal{D}_{B\rightarrow A^*}(\rho_{AB})=\proj{\phi}_{AA^*}$ 
and $\ket{\phi}_{AA^*}=U_A d_A^{-1/2}\sum_{j\in [d_A]}\ket{j}_A\otimes\ket{j}_{A^*}$~\cite{schumacher1996quantum}. 
This equivalence is robust and systematic constructions for $\mathcal{D}_{B\rightarrow A^*}$ have been proposed~\cite{schumacher1996quantum,schumacher2001approximate}. 
Since $\proj{\phi}_{AA^*}$ is a pure, maximally entangled state, these results establish an equivalence between the maximality of the mutual information and maximal entanglement. 
\end{rem}

By combining Corollary~\ref{cor:equivalence}~(c) with Remark~\ref{rem:max-mi}, we can answer Question~1 in the affirmative.
\begin{description}
    \item[Answer to Question~1] The maximality of any R\'enyi reflected entropy of order $\alpha\in (0,\infty]$ corresponds to maximal entanglement due to Corollary~\ref{cor:equivalence}~(c) and Remark~\ref{rem:max-mi}.
\end{description}

\subsection{Comparison of (R\'enyi) reflected entropy with entanglement of purification}
\label{ssec:comparison_eop}

In holography, the entanglement of purification is often compared to half the reflected entropy, as both measures have been proposed as candidates for the boundary quantity dual to the area of the entanglement wedge cross section (cf. Section~\ref{sec:introduction}). 
Given such comparisons between $E_P(A:B)_\rho$ and $\frac{1}{2}S_R(A:B)_\rho$, it is worth emphasizing that these two measures can exhibit drastically different behavior for general quantum states.
In particular, the entanglement of purification does \emph{not} satisfy an analogous equivalence to that stated in Corollary~\ref{cor:equivalence}~(c) for the R\'enyi reflected entropy. 
Although it is true that maximal mutual information implies maximal entanglement of purification (i.e., if $I(A:B)_\rho=2\log \min\{d_A,d_B\}$, then $E_P(A:B)_\rho=\log \min\{d_A,d_B\}$ due to~\eqref{eq:eop-bounds}), 
the converse is false. 
A simple counterexample is given in Appendix~\ref{app:eop}.

In recent work~\cite{akers2023entanglement} (see also~\cite{couch2024possibility,chen2025renyientanglementpurificationhalf}), the following upper and lower bounds on the entanglement of purification of $\rho_{AB}\in \mathcal{S}(AB)$ have been derived:
\begin{align}
\frac{1}{2}S_R^{(2)}(A:B)_\rho\leq E_P(A:B)_\rho \leq S_R(A:B)_\rho .
\end{align}
Interestingly, we have shown in Corollary~\ref{cor:bounds} that the min-reflected entropy has the same upper and lower bounds. 
However, the entanglement of purification cannot generally be inserted anywhere in the chain of inequalities in Corollary~\ref{cor:bounds}, as shown in Appendix~\ref{app:eop2}.

\subsection{Operational interpretation of min-reflected entropy}\label{ssec:operational}
Since the doubly minimized PRMI of order $1/2$ has an operational interpretation~\cite{burri2024doublyminimizedpetzrenyi}, not only properties of the doubly minimized PRMI of order $1/2$ but also this operational interpretation can be carried over to the min-reflected entropy due to our main result. 

Furthermore, it is known that the singly minimized PRMI of order $1/2$ quantifies the entanglement fidelity of the Petz decoder (i.e., the decoder induced by the Petz recovery map) for one-shot entanglement transmission~\cite{burri2024entanglementfidelitypetzdecoder}. 
Since the doubly minimized PRMI of order $1/2$ is a natural lower bound for the singly minimized PRMI of order $1/2$, our main result thus leads to a relation between the min-reflected entropy and Petz recovery. 

Building on these considerations, we provide the following answer to Question~2.
\begin{description}
    \item[Answer to Question~2] The min-reflected entropy has an operational interpretation (if certain minor technical conditions~\cite[Corollary~11, Remark~6]{burri2024doublyminimizedpetzrenyi} on $\rho_{AB}$ are fulfilled): 
    It equals the forward $\beta$-cutoff rate as $\beta \rightarrow -1$ from above in the hypothesis testing problem whose null hypothesis is given by $\rho_{AB}^{\otimes n}$ and whose alternative hypothesis consists of permutation invariant product states of the form $\sigma_{A^n}\otimes\tau_{B^n}$ (or product states of the form $\sigma_A^{\otimes n}\otimes \tau_B^{\otimes n}$) in the asymptotic limit as $n\rightarrow\infty$~\cite{burri2024doublyminimizedpetzrenyi}. 
    More formally: For any $\rho_{AB}\in \mathcal{S}(AB)$ such that $I_{1/2}^{\downarrow\downarrow}(A:B)_\rho\neq I(A:B)_\rho$ holds
    \begin{align}\label{eq:r0f}
    \lim_{\beta \rightarrow -1^+}R_0^{(f)}(\beta)_\rho
    =I_{1/2}^{\downarrow\downarrow}(A:B)_\rho
    =S_R^{(\infty)}(A:B)_\rho,
    \end{align}
    where $R_0^{(f)}(\beta)_\rho$ denotes the forward $\beta$-cutoff rate. 
    The first equality in~\eqref{eq:r0f} follows from~\cite[Corollary~11]{burri2024doublyminimizedpetzrenyi}, and the second equality follows from our main result. 
    Thus, the min-reflected entropy is linked to a specific type of correlation detection task. 
    
    Furthermore, there is a deeper relation between the min-reflected entropy and Petz recovery: 
    For any unit vector $\ket{\rho}_{RA}\in RA$ and any $\mathcal{N}_{A\rightarrow B}\in \CPTP(A,B)$
    \begin{align}
    \log F_e (\rho_A,\mathcal{D}_{B\rightarrow A}^{\mathrm{P}}\circ\mathcal{N}_{A\rightarrow B})
    &=- I_{1/2}^{\uparrow\downarrow}(R:E)_{\mathcal{N}_{A\rightarrow E}^c(\proj{\rho}_{RA})}
    \label{eq:log-fe1}\\
    &\leq -I_{1/2}^{\downarrow\downarrow}(R:E)_{\mathcal{N}_{A\rightarrow E}^c(\proj{\rho}_{RA})}
    =-S_R^{(\infty)}(R:E)_{\mathcal{N}_{A\rightarrow E}^c(\proj{\rho}_{RA})}
    \label{eq:log-fe2}
\end{align}
	where 
	$F_e(\rho_A,\mathcal{D}_{B\rightarrow A}^{\mathrm{P}}\circ\mathcal{N}_{A\rightarrow B})$ denotes the entanglement fidelity of the Petz decoder $\mathcal{D}_{B\rightarrow A}^{\mathrm{P}}$ for one-shot entanglement transmission of $\rho_{A}$ over the noisy channel $\mathcal{N}_{A\rightarrow B}$, 
	and $\mathcal{N}_{A\rightarrow E}^c$ denotes a complementary channel to $\mathcal{N}_{A\rightarrow B}$~\cite{burri2024entanglementfidelitypetzdecoder}. 
	\eqref{eq:log-fe1} follows from~\cite{burri2024entanglementfidelitypetzdecoder}, and the equality in~\eqref{eq:log-fe2} follows from our main result. 
	Thus, the entanglement fidelity of the Petz decoder is upper bounded in terms of the min-reflected entropy associated with a complementary channel. 
\end{description}
The second part of this answer offers insight into a question raised in previous work~\cite{akers2022page} regarding the existence of a deeper connection between information recovery via the Petz map and the reflected entropy. 
It demonstrates that a modified variant of this conjecture holds, where the reflected entropy is replaced by the min-reflected entropy.

\subsection{Comparison with classical case}\label{ssec:classical}
Our main result yields two new methods for computing the doubly minimized PRMI of order $1/2$ for any $\rho_{AB}\in \mathcal{S}(AB)$, as indicated by the right-hand sides of the following two lines.
\begin{align}
I_{1/2}^{\downarrow\downarrow}(A:B)_\rho
=S_R^{(\infty)}(A:B)_\rho
&=H_{\infty}(AA^*)_{\hat{\rho}}
=-\log \max ( \spec(\hat{\rho}_{AA^*}))
\label{eq:i12-maxspec}\\
&=\lim_{n\rightarrow\infty}S_R^{(n)}(A:B)_\rho
=\lim_{n\rightarrow\infty}H_n(AA^*)_{\hat{\rho}}
\label{eq:i12-limit}
\end{align}
Here, $\hat{\rho}_{AA^*}\coloneqq \tr_{BB^*}[\proj{\hat{\rho}}_{ABA^*B^*}]$ denotes the marginal state on $AA^*$ of the canonical purification $\ket{\hat{\rho}}_{ABA^*B^*}$ of $\rho_{AB}$. 
Note that the expressions in~\eqref{eq:i12-maxspec} and~\eqref{eq:i12-limit} are not closed-form expressions for $I_{1/2}^{\downarrow\downarrow}(A:B)_\rho$, 
because~\eqref{eq:i12-maxspec} involves a largest eigenvalue problem, 
and~\eqref{eq:i12-limit} involves a limit. 
These expressions are nonetheless useful, because an evaluation of $I_{1/2}^{\downarrow\downarrow}(A:B)_\rho$ via~\eqref{eq:i12-maxspec} or~\eqref{eq:i12-limit} is typically easier compared to other available alternatives, such as a definition-based evaluation of $I_{1/2}^{\downarrow\downarrow}(A:B)_\rho$ or an evaluation by means of universal permutation invariant states~\cite{burri2024doublyminimizedpetzrenyi}. 
In particular, since the replica trick is a widely used method for calculating the R\'enyi reflected entropy $S_R^{(n)}(A:B)_\rho$ for $n\in \mathbb{N}_{\geq 2}$, evaluating $I_{1/2}^{\downarrow\downarrow}(A:B)_\rho$ via~\eqref{eq:i12-limit} is feasible according to current practice in the corresponding settings.

In order to compare our main result with a similar result in classical information theory~\cite{lapidoth2019two}, let us consider the evaluation of our main result (Theorem~\ref{thm:srinf-i12}) for a CC state. 
\eqref{eq:reflected-cc} then gives rise to the following corollary of Theorem~\ref{thm:srinf-i12}.
\begin{cor}[CC states]\label{cor:cc}
Let $P_{XY}$ be the joint PMF of two random variables $X,Y$ over $[d_A], [d_B]$. 
Let $M_{x,y}\coloneqq \sqrt{P_{XY}(x,y)}$ for all $x\in [d_A],y\in [d_B]$ 
and let $M\coloneqq (M_{x,y})_{x\in [d_A],y\in [d_B]}$ be a $(d_A\times d_B)$-matrix. 
Let $M^T$ denote the transpose of $M$. 
Let $\{\ket{a_x}_A\}_{x\in [d_A]},\{\ket{b_y}_B\}_{y\in [d_B]}$ be orthonormal bases for $A,B$, 
and let 
$\rho_{AB}\coloneqq \sum_{x\in [d_A],y\in [d_B]} P_{XY}(x,y)\proj{a_x,b_y}_{AB}$. 
Then, 
\begin{align}\label{eq:cc-states}
I_{1/2}^{\downarrow\downarrow}(A:B)_\rho
=S_R^{(\infty)}(A:B)_\rho 
=-2\log \lVert M\rVert_\infty 
&=-2\log \sigma_{\max}(M)
\\
&=-2\log \sqrt{\max(\spec(M^TM))}
\end{align}
where $\sigma_{\max}(M)$ denotes the largest singular value of $M$.
\end{cor}

The state $\rho_{AB}$ in Corollary~\ref{cor:cc} is a CC state. 
Hence, the minimization over $(\sigma_A,\tau_B)\in \mathcal{S}(A)\times \mathcal{S}(B)$ occurring in the definition of $I_{1/2}^{\downarrow\downarrow}(A:B)_\rho$ can be restricted to states that are diagonal in the orthonormal bases $\{\ket{a_x}_A\}_{x\in [d_A]},\{\ket{b_y}_B\}_{y\in [d_B]}$~\cite{burri2024doublyminimizedpetzrenyi}. 
Thus, Corollary~\ref{cor:cc} implies the following corollary.
\begin{cor}[Classical case]\label{cor:classical}
Let $P_{XY}$ be the joint PMF of two random variables $X,Y$ over finite alphabets $\mathcal{X},\mathcal{Y}$. 
Let $M_{x,y}\coloneqq \sqrt{P_{XY}(x,y)}$ for all $x\in \mathcal{X},y\in \mathcal{Y}$ 
and let $M\coloneqq (M_{x,y})_{x\in \mathcal{X} ,y\in \mathcal{Y}}$ be a $(\lvert\mathcal{X}\rvert \times \lvert \mathcal{Y}\rvert)$-matrix. 
Let $M^T$ denote the transpose of $M$. 
Then, 
\begin{align}
\inf_{Q_X,R_Y}
D_{1/2}(P_{XY}\| Q_XR_Y )
&=-2\log \sigma_{\max}(M)
\label{eq:classical}\\
&=-2\log \sqrt{\max(\spec(M^TM))}
\label{eq:classical2}
\end{align}
where the minimization is over PMFs $Q_X,R_Y$, 
and $\sigma_{\max}(M)$ denotes the largest singular value of $M$.
\end{cor}
The identity in~\eqref{eq:classical} has been derived previously in classical information theory~\cite[Lemma~6]{lapidoth2019two}. 
It shows that the classical analogue of $I_{1/2}^{\downarrow\downarrow}(A:B)_\rho$ admits a largest singular value formulation, which can also be interpreted as a largest eigenvalue formulation due to~\eqref{eq:classical2}. 
In view of Corollaries~\ref{cor:cc} and~\ref{cor:classical}, our main result in Theorem~\ref{thm:srinf-i12} can be seen as a generalization of~\cite[Lemma~6]{lapidoth2019two} from CC states to general quantum states, because Theorem~\ref{thm:srinf-i12} implies that $I_{1/2}^{\downarrow\downarrow}(A:B)_\rho$ admits a largest eigenvalue formulation for any $\rho_{AB}\in \mathcal{S}(AB)$ and not just for CC states, see~\eqref{eq:i12-maxspec}. 
We can therefore answer Question~3 as follows.
\begin{description}
    \item[Answer to Question~3] 
    According to~\eqref{eq:i12-maxspec}, $I_{1/2}^{\downarrow\downarrow}(A:B)_\rho$ is determined by the largest eigenvalue of 
    $\hat{\rho}_{AA^*}\coloneqq \tr_{BB^*}[\proj{\hat{\rho}}_{ABA^*B^*}]$, where $\ket{\hat{\rho}}_{ABA^*B^*}$ denotes the canonical purification of $\rho_{AB}$.  
    This provides a largest eigenvalue formulation for the doubly minimized PRMI of order $1/2$ that holds for all bipartite quantum states $\rho_{AB}$,  
    generalizing a previous result for the classical case~\cite[Lemma~6]{lapidoth2019two} to the quantum domain.
\end{description}

\section{Reconsideration of definition of R\'enyi reflected entropy}\label{sec:definition}

As explained in Section~\ref{ssec:renyi-reflected}, the definition of the R\'enyi reflected entropy refers to the \emph{canonical purification} of the corresponding state. 
Although the canonical purification is well-defined, it is somewhat unclear what its distinguishing features are from a formal perspective. 
In fact, the main obvious properties of the canonical purification do not uniquely specify it but rather characterize a larger set of purifications, 
in this work referred to as \emph{$\mathcal{CF}$-invariant purifications}, 
which themselves form a subset of \emph{marginal-canonical purifications}. 
These notions will be introduced in Section~\ref{ssec:purification}.

This naturally raises the question of whether the $n$-R\'enyi reflected entropy nonetheless depends in a crucial but obscure way on the canonical purification, or whether these larger sets of purifications are also admissible. 
In Section~\ref{ssec:deflected}, we show that the latter is the case for all $n\in \mathbb{N}_{\geq 2}\cup \{\infty\}$. 
Specifically, we will show that the $n$-R\'enyi reflected entropy can be expressed as a minimization over all marginal-canonical purifications (Theorem~\ref{thm:minimized}), where no explicit reference to the canonical purification remains.

\subsection{$\mathcal{CF}$-invariant and marginal-canonical purification}\label{ssec:purification}

The canonical purification of $\rho_A\in \mathcal{S}(A)$ is defined as $\ket{\hat{\rho}}_{AA^*}\coloneqq \sqrt{\rho_A}\ket{\Omega}_{AA^*}$, see~\eqref{eq:def-cp}. 
This immediately implies that the canonical purification has the following main properties.
\begin{enumerate}[label=(\Alph*)]
\item \emph{Marginal state on $A$:} 
$\tr_{A^*}[\proj{\hat{\rho}}_{AA^*}]=\rho_{A}$
\item \emph{Marginal state on $A^*$:} 
$\tr_{A}[\proj{\hat{\rho}}_{AA^*}]=\rho_{A^*}^*$
\item \emph{$\mathcal{CF}$-invariance:} 
$(\swapf \ket{\hat{\rho}}_{AA^*})^*=\ket{\hat{\rho}}_{AA^*}$
\end{enumerate}

These properties do not uniquely specify the canonical purification. 
To formalize this observation, we introduce the following two notions. 
They correspond to vectors satisfying $(A)\land (C)$ and $(A)\land (B)$ respectively. 
(Note that $(A)\land (C)$ implies $(B)$, as we will see in Proposition~\ref{prop:invariant}~(f) below.)

\begin{defn}[$\mathcal{CF}$-invariant and marginal-canonical purification]\label{def:marginal}
Let $\rho_A\in \mathcal{S}(A)$, $\ket{\bar{\rho}}_{AA^*}\in AA^*$. 
$\ket{\bar{\rho}}_{AA^*}$ is a \emph{$\mathcal{CF}$-invariant purification of $\rho_A$} 
if $\tr_{A^*}[\proj{\bar{\rho}}_{AA^*}]=\rho_A$ and $(\swapf \ket{\bar{\rho}}_{AA^*})^*=\ket{\bar{\rho}}_{AA^*}$. 
$\ket{\bar{\rho}}_{AA^*}$ is a \emph{marginal-canonical purification of $\rho_A$} 
if $\tr_{A^*}[\proj{\bar{\rho}}_{AA^*}]=\rho_A$ and $\tr_{A}[\proj{\bar{\rho}}_{AA^*}]=\rho^*_{A^*}$.
\end{defn}

The following proposition lists several properties of the three types of purifications defined above. 
Since their proof is straightforward, the proposition is stated without proof. 

\begin{prop}[Canonical, $\mathcal{CF}$-invariant, and marginal-canonical purification]\label{prop:invariant}
Let $\rho_{A}\in \mathcal{S}(A)$, $\ket{\bar{\rho}}_{AA^*}\in AA^*$.
Then all of the following hold.
\begin{enumerate}[label=(\alph*)]
\item $\ket{\bar{\rho}}_{AA^*}$ is the canonical purification of $\rho_{A}$ 
iff there exist an orthonormal basis $\{\ket{e_j}_{A}\}_{j\in [d_A]}$ for $A$ and 
$(p_j)_{j\in [d_A]}\in [0,1]^{\times d_A}$ such that 
$\rho_{A}=\sum_{j\in [d_A]}p_j\proj{e_j}_{A}$ and 
\begin{align}
\ket{\bar{\rho}}_{AA^*}=\sum_{j\in [d_A]}\sqrt{p_j}\ket{e_j}_{A}\otimes \ket{e_j}_{A^*}^*.
\end{align}
\item $\ket{\bar{\rho}}_{AA^*}$ is a $\mathcal{CF}$-invariant purification of $\rho_{A}$ 
iff there exists $U_{A}\in \mathcal{U}(A)$ such that $\ket{\bar{\rho}}_{AA^*}=\sqrt{\rho_{A}}U_A\ket{\Omega}_{AA^*}$ and $[\rho_A,U_A]=0$ and $U_{A}^2=1_{A}$.
\item $\ket{\bar{\rho}}_{AA^*}$ is a $\mathcal{CF}$-invariant purification of $\rho_{A}$ 
iff there exist an orthonormal basis $\{\ket{e_j}_{A}\}_{j\in [d_A]}$ for $A$, 
$(p_j)_{j\in [d_A]}\in [0,1]^{\times d_A}$, 
and $(b_j)_{j\in [d_A]}\in \{0,1\}^{\times d_A}$ such that 
$\rho_{A}=\sum_{j\in [d_A]}p_j\proj{e_j}_{A}$ and
\begin{align}\label{eq:structure-invariant}
\ket{\bar{\rho}}_{AA^*}&= \sum_{j\in [d_A]}(-1)^{b_j} \sqrt{p_j} \ket{e_j}_{A}\otimes \ket{e_j}_{A^*}^* .
\end{align}
\item $\ket{\bar{\rho}}_{AA^*}$ is a marginal-canonical purification of $\rho_{A}$ 
iff there exists $U_{A}\in \mathcal{U}(A)$ such that 
$\ket{\bar{\rho}}_{AA^*}=\sqrt{\rho_{A}}U_A\ket{\Omega}_{AA^*}$ and $[\rho_A,U_A]=0$.
\item \sloppy 
$\ket{\bar{\rho}}_{AA^*}$ is a marginal-canonical purification of $\rho_{A}$ 
iff there exist an orthonormal basis $\{\ket{e_j}_{A}\}_{j\in [d_A]}$ for $A$, 
$(p_j)_{j\in [d_A]}\in [0,1]^{\times d_A}$, 
and $(\theta_j)_{j\in [d_A]}\in [0,2\pi)^{\times d_A}$ such that 
$\rho_{A}=\sum_{j\in [d_A]}p_j\proj{e_j}_{A}$ and
\begin{align}\label{eq:structure-partially-invariant}
\ket{\bar{\rho}}_{AA^*}= \sum_{j\in [d_A]}e^{i\theta_j} \sqrt{p_j} \ket{e_j}_{A}\otimes \ket{e_j}_{A^*}^* .
\end{align}
\item If $\ket{\bar{\rho}}_{AA^*}$ is the canonical purification of $\rho_A$, then $\ket{\bar{\rho}}_{AA^*}$ is a $\mathcal{CF}$-invariant purification of $\rho_A$. 
If $\ket{\bar{\rho}}_{AA^*}$ is a $\mathcal{CF}$-invariant purification of $\rho_A$, then 
$\ket{\bar{\rho}}_{AA^*}$ is a marginal-canonical purification of $\rho_A$.
\end{enumerate}
\end{prop}

By definition, the canonical purification of $\rho_A$ is unique, unlike the other two types of purification. 
Proposition~\ref{prop:invariant}~(c) shows that $\mathcal{CF}$-invariant purifications retain a residual degree of freedom associated with the binary choice of certain signs. 
Similarly, Proposition~\ref{prop:invariant}~(e) shows that marginal-canonical purifications retain a residual degree of freedom associated with the choice of certain complex phases. 
Proposition~\ref{prop:invariant}~(f) implies that the three types of purifications are ordered by inclusion, as illustrated in Figure~\ref{fig:purifications}.

\begin{figure}
\begin{overpic}[width=\textwidth]{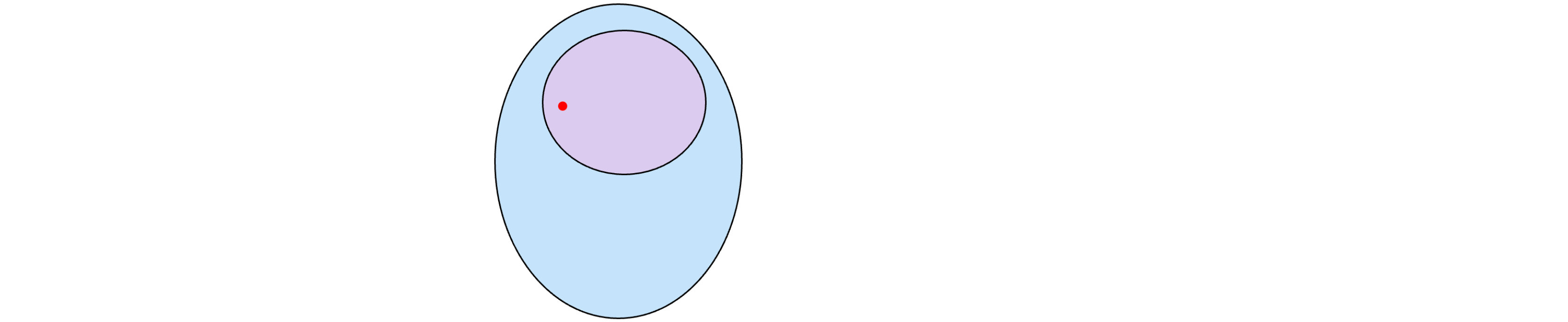}
\begin{scriptsize}
  \put(53,8){\textcolor{NavyBlue}{\textbf{Marginal-canonical purification of $\rho_A$}}}
  \put(53,5){\textcolor{NavyBlue}{Definition: $\tr_{A^*}[\proj{\bar{\rho}}_{AA^*}]=\rho_A,\tr_{A}[\proj{\bar{\rho}}_{AA^*}]=\rho_{A^*}^*$}}
  \put(53,2){\textcolor{NavyBlue}{$\Leftrightarrow\ket{\bar{\rho}}_{AA^*}=\sum\limits_{j\in [d_A]}e^{i\theta_j}\sqrt{p_j}\ket{e_j}_A\otimes\ket{e_j}_{A^*}^*,\quad \theta_j\in [0,2\pi)$}}
  \put(53,21){\textcolor{Purple}{\textbf{$\mathcal{CF}$-invariant purification of $\rho_A$}}}
  \put(53,18){\textcolor{Purple}{Definition: $\tr_{A^*}[\proj{\bar{\rho}}_{AA^*}]=\rho_A,(\mathcal{F}\ket{\bar{\rho}}_{AA^*})^*=\ket{\bar{\rho}}_{AA^*}$}}
  \put(53,15){\textcolor{Purple}{$\Leftrightarrow\ket{\bar{\rho}}_{AA^*}=\sum\limits_{j\in [d_A]}(-1)^{b_j}\sqrt{p_j}\ket{e_j}_A\otimes\ket{e_j}_{A^*}^*, \quad b_j\in \{0,1\}$}}
  \put(0,14){\textcolor{red!85!black}{\textbf{Canonical purification of $\rho_A$}}}
  \put(0,11){\textcolor{red!85!black}{Definition: $\ket{\bar{\rho}}_{AA^*}=\sqrt{\rho_A}\ket{\Omega}_{AA^*}$}}
  \put(0,8){\textcolor{red!85!black}{$\Leftrightarrow\ket{\bar{\rho}}_{AA^*}=\sum\limits_{j\in [d_A]}\sqrt{p_j}\ket{e_j}_A\otimes\ket{e_j}_{A^*}^*$}}
\end{scriptsize}
\end{overpic}
\caption{Illustration of Proposition~\ref{prop:invariant}. 
According to Proposition~\ref{prop:invariant}~(f), the canonical purification of $\rho_A$ (\emph{red}) is an element of the set of $\mathcal{CF}$-invariant purifications of $\rho_A$ (\emph{purple}), which is a subset of the set of marginal-canonical purifications of $\rho_A$ (\emph{blue}).}
\label{fig:purifications}
\end{figure}

\subsection{R\'enyi reflected entropy from marginal-canonical purification}
\label{ssec:deflected}

\begin{defn}[Minimized R\'enyi reflected entropy]\label{def:minimized}
The \emph{minimized R\'enyi reflected entropy of order $\alpha\in [0,\infty]$} is defined for $\rho_{AB}\in \mathcal{S}(AB)$ as 
\begin{align}\label{eq:def-m1}
S_{R}^{\downarrow (\alpha)}(A:B)_\rho 
&\coloneqq \inf_{\substack{\ket{\bar{\rho}}_{ABA^*B^*}\in ABA^*B^*: \\ \tr_{A^*B^*}[\proj{\bar{\rho}}_{ABA^*B^*}]=\rho_{AB},\\ \tr_{AB}[\proj{\bar{\rho}}_{ABA^*B^*}]=\rho_{A^*B^*}^*}} H_\alpha (AA^*)_{\bar{\rho}}
\end{align}
where $\bar{\rho}_{AA^*}\coloneqq \tr_{BB^*}[\proj{\bar{\rho}}_{ABA^*B^*}]$.
\end{defn}

The minimization over all marginal-canonical purifications of $\rho_{AB}$ in~\eqref{eq:def-m1} can also be formulated as a minimization over unitaries $U_{AB}$ that commute with $\rho_{AB}$ due to Proposition~\ref{prop:invariant}~(d). 
Thus, we have for any $\rho_{AB}\in \mathcal{S}(AB)$ and $\alpha\in [0,\infty]$
\begin{align}\label{eq:def-m2}
S_{R}^{\downarrow (\alpha)}(A:B)_\rho 
&= \inf_{\substack{U_{AB}\in \mathcal{U}(AB):\\ [\rho_{AB},U_{AB}]=0}} H_\alpha (AA^*)_{\bar{\rho}(U)}
\end{align}
where $\bar{\rho}(U)_{AA^*}\coloneqq \tr_{BB^*}[\proj{\bar{\rho}(U)}_{ABA^*B^*}]$ 
and $\ket{\bar{\rho}(U)}_{ABA^*B^*}\coloneqq \sqrt{\rho_{AB}}U_{AB}\ket{\Omega}_{AA^*}\otimes\ket{\Omega}_{BB^*}$.

Although the \emph{minimized R\'enyi reflected entropy} is newly introduced in this work, it has clear connections to previously defined measures. 
Firstly, it is upper bounded by the R\'enyi reflected entropy because $U_{AB}\coloneqq 1_{AB}$ is in the feasible set of~\eqref{eq:def-m2}. 
Secondly, it is even upper bounded by any \emph{R\'enyi $s$-deflected entropy}~\cite{dutta2023reflected,berthiere2023markov}, which is defined as follows. 

\begin{defn}[R\'enyi $s$-deflected entropy]\label{def:deflected}
The \emph{R\'enyi $s$-deflected entropy of order $\alpha\in [0,\infty]$ for $s\in \mathbb{R}$} is defined for $\rho_{AB}\in \mathcal{S}(AB)$ as~\cite{dutta2023reflected}
\begin{align}
S_{D,s}^{(\alpha)}(A:B)_\rho 
&\coloneqq H_\alpha (AA^*)_{\bar{\rho}^{(s)}}
\end{align}
where $\bar{\rho}^{(s)}_{AA^*}\coloneqq \tr_{BB^*}[\proj{\bar{\rho}^{(s)}}_{ABA^*B^*}]$ 
and $\ket{\bar{\rho}^{(s)}}_{ABA^*B^*}\coloneqq \rho_{AB}^{\frac{1}{2}+is}\ket{\Omega}_{AA^*}\otimes\ket{\Omega}_{BB^*}$.
\end{defn}
Since $U_{AB}\coloneqq \rho_{AB}^{is}$ is an isometry whose support coincides with $\supp(\rho_{AB})$, it is immediate to infer from~\eqref{eq:def-m2} and Definition~\ref{def:deflected} that the minimized R\'enyi reflected entropy is indeed upper bounded by the R\'enyi $s$-deflected entropy for any $s\in \mathbb{R}$. 
Consequently, we have for all $\alpha\in [0,\infty],\rho_{AB}\in \mathcal{S}(AB)$
\begin{align}\label{eq:sr-ineq}
S_{R}^{\downarrow (\alpha)}(A:B)_\rho 
&\leq \inf_{s\in \mathbb{R}} S_{D,s}^{(\alpha)}(A:B)_\rho  
\leq S_{D,0}^{(\alpha)}(A:B)_\rho  
=S_{R}^{(\alpha)}(A:B)_\rho .
\end{align}
All of these inequalities are saturated for $\alpha\in \mathbb{N}_{\geq 2}\cup \{\infty\}$, as shown in the following theorem. 
The proof of this theorem is given in Appendix~\ref{app:proof_minimized}.

\begin{thm}[Minimized R\'enyi reflected entropy equals R\'enyi reflected entropy]\label{thm:minimized}
Let $\rho_{AB}\in \mathcal{S}(AB),n\in \mathbb{N}_{\geq 2}\cup \{\infty\}$. Then
\begin{align}
S_{R}^{\downarrow (n)}(A:B)_\rho 
=\min_{s\in \mathbb{R}} S_{D,s}^{(n)}(A:B)_\rho 
=S_{R}^{(n)}(A:B)_\rho .
\label{eq:ref-min}
\end{align}
\end{thm}

\begin{rem}[Range of validity of Theorem~\ref{thm:minimized}]
\label{rem:02}
The assertion in~\eqref{eq:ref-min} can be violated for $n\in [0,2)$. 
An example of a state $\rho_{AB}$ such that $S_R^{\downarrow(n)}(A:B)_\rho<S_R^{(n)}(A:B)_\rho$ for all $n\in [0,2)$ is given in Appendix~\ref{app:ex-02}.
\end{rem}

Theorem~\ref{thm:minimized} implies that in order to calculate the R\'enyi reflected entropy of order $n\in \mathbb{N}_{\geq 2}\cup \{\infty\}$, it is not required to evaluate $H_n (AA^*)_{\bar{\rho}}$ for the marginal state on $AA^*$ of the \emph{canonical purification} of $\rho_{AB}$. 
According to Theorem~\ref{thm:minimized}, it can also be obtained by minimizing $H_n (AA^*)_{\bar{\rho}}$ over all \emph{marginal-canonical purifications} of $\rho_{AB}$. 
While the canonical purification of $\rho_{AB}$ is a peculiar unique pure state which is not uniquely specified by its main properties, 
marginal-canonical purifications of $\rho_{AB}$ can be more clearly identified as the vectors with two specific properties, see~\eqref{eq:def-m1}. 
This offers a qualitatively novel perspective on the $n$-R\'enyi reflected entropy, replacing its obscure dependence on the canonical purification with an optimization over a larger class of purifications that is constrained by two simple properties.

Note that Theorem~\ref{thm:minimized} holds as well if the minimization occurring in the definition of the minimized R\'enyi reflected entropy is restricted to $\mathcal{CF}$-invariant purifications. 
This follows from Theorem~\ref{thm:minimized} since $\mathcal{CF}$-invariant purifications form a subset of all marginal-canonical purifications and this subset includes the canonical purification, see Proposition~\ref{prop:invariant}~(f).

\section{Conclusion}\label{sec:conclusion}
The main technical contribution of this work is a proof that two correlation measures previously studied independently are, in fact, identical: 
The min-reflected entropy is equal to the doubly minimized PRMI of order $1/2$ (Theorem~\ref{thm:srinf-i12}). 

Implications of this main result have been explored in Section~\ref{sec:implications}. 
They include new bounds on the min-reflected entropy that are analogous to previously derived bounds for the entanglement of purification~(\ref{ssec:bounds}, \ref{ssec:comparison_eop}), 
the correspondence between the maximality of R\'enyi reflected entropies and maximal entanglement~(\ref{ssec:max}), 
an operational interpretation of the min-reflected entropy based on a correlation detection task as well as a connection between the min-reflected entropy and information recovery via the Petz map~(\ref{ssec:operational}), 
and a rederivation of a previous result in classical information theory~\cite[Lemma~6]{lapidoth2019two}, suggesting that our main result is a generalization of their result from the classical to the quantum domain~(\ref{ssec:classical}). 
These implications have enabled us to answer several previously open questions posed in the introduction (\ref{ssec:max}, \ref{ssec:operational}, \ref{ssec:classical}).

As a side result, we have shown that the R\'enyi reflected entropy of order $n\in \mathbb{N}_{\geq 2}\cup \{\infty\}$ can also be expressed as a minimization over larger classes of purifications -- marginal-canonical or $\mathcal{CF}$-invariant purifications -- where no explicit reference to the canonical purification remains (Theorem~\ref{thm:minimized}). 
This provides a qualitatively novel perspective on the R\'enyi reflected entropy which may be of independent interest. 
For instance, in a holographic context, it would be interesting to explore whether the use of marginal-canonical or $\mathcal{CF}$-invariant purifications on the boundary is related to the imposition of CPT symmetry in the bulk~\cite{dutta2019canonical,engelhardt2019coarse,engelhardt2018decoding}. 
Since this side result also concerns the min-reflected entropy ($n=\infty$), it also applies to our main result, offering an alternative interpretation.

As noted in the introduction, previous works on R\'enyi reflected entropies have primarily focused on the reflected entropy. 
Most studies that involve $(m,n)$-R\'enyi reflected entropies aim to ultimately compute the reflected entropy through analytic continuation in $m$ and $n$, followed by taking the limit $n\rightarrow 1,m\rightarrow 1$. 
However, the reflected entropy is not a correlation measure for general quantum states~\cite{hayden2023reflected} and lacks an operational interpretation, casting doubt on its overall relevance. 
Our main result suggests a change in focus, especially in settings where the reflected entropy lacks fundamental relevance. 
Our main result implies that taking the limit $n\rightarrow\infty,m\rightarrow 1$ instead yields a correlation measure that is natural from an information-theoretic point of view for general quantum states: 
the doubly minimized PRMI of order $1/2$. 

The implications of our work for the reflected entropy are less clear. 
We leave it as an open question to understand whether there is a deeper link between the reflected entropy and the min-reflected entropy in settings where the reflected entropy appears to have fundamental relevance, such as in holographic theories.

\begin{acknowledgments}
I would like to thank Renato Renner for his valuable feedback on a draft of this work. 
I also thank Giulia Mazzola and David Sutter for discussions. 
This work was supported by the Swiss National Science Foundation 
via project No.\ 20QU-1\_225171 
and the National Centre of Competence in Research SwissMAP, 
and the Quantum Center at ETH Zurich.
\end{acknowledgments}

\appendix
\section{Invariance under complex conjugation and swapping}\label{app:invariance}
In this section, we develop basic techniques for handling objects that remain invariant under the joint application of complex conjugation $(\mathcal{C}\equiv \, ^*)$ and swapping $(\mathcal{F})$. 
We define the notion of $\mathcal{CF}$-invariance for the following two types of objects.
\begin{defn}[$\mathcal{CF}$-invariance]
$\ket{\phi}\in AA^*$ is \emph{$\mathcal{CF}$-invariant} if 
$(\swapf\ket{\phi})^*=\ket{\phi}$. 
$X\in \mathcal{L}(AA^*)$ is \emph{$\mathcal{CF}$-invariant} if 
$(X^{\swapf})^*=X$.
\end{defn}
\begin{rem}[Notation for action of $\mathcal{CF}$]
To simplify the notation, we will occasionally write $X^{\swapf *}$ instead of $(X^{\swapf})^*$ for $X\in \mathcal{L}(AA^*)$.
\end{rem}

The following proposition lists basic properties of the joint application of complex conjugation and swapping. 
Since they follow directly from the definitions of the swap operator and complex conjugation in  Section~\ref{ssec:swap}, the proposition is presented without proof. 
The properties listed in Proposition~\ref{prop:basic} will be used extensively in the rest of Appendix~\ref{app:invariance}.

\begin{prop}[Properties of $\mathcal{CF}$]\label{prop:basic}
Let $\ket{\phi}\in AA^*,X\in \mathcal{L}(AA^*)$.
Then all of the following hold.
\begin{enumerate}[label=(\alph*)]
\item \emph{$\mathcal{C}$ and $\mathcal{F}$ commute:}
$(\swapf \ket{\phi})^*=\swapf (\ket{\phi}^*)$ 
and $(X^{\swapf})^*=(X^*)^{\swapf }$.
\item \emph{$\mathcal{CF}$ is an involution:}
$(\swapf (\swapf \ket{\phi})^*)^*=\ket{\phi}$ and 
$(X^{\swapf *})^{\swapf *}=X$.
\item \emph{$\mathcal{CF}$ is antilinear:} 
For all $\alpha,\beta\in \mathbb{C},\ket{\psi}\in AA^*,Y\in \mathcal{L}(AA^*)$
\begin{align}
(\swapf  (\alpha\ket{\phi}+\beta\ket{\psi}))^*
&=\alpha^*(\swapf  \ket{\phi})^*+\beta^*(\swapf  \ket{\psi})^*,
\\
(\swapf  (\alpha X+\beta Y))^*
&=\alpha^*(\swapf  X)^*+\beta^*(\swapf Y)^* .
\end{align}
\item \emph{$\mathcal{CF}$ preserves rank:} 
$\rank(X^{\swapf *})=\rank(X)$.
\item \emph{$\mathcal{CF}$ preserves positive semi-definiteness:}
If $X\geq 0$, then $X^{\swapf *}\geq 0$.
\item \emph{$\mathcal{CF}$-invariant unit vectors induce $\mathcal{CF}$-invariant projections:}
If $\ket{\phi}\in AA^*$ is a $\mathcal{CF}$-invariant unit vector, then $\proj{\phi}\in \Proj_1(AA^*)\subseteq \mathcal{L}(AA^*)$ is $\mathcal{CF}$-invariant.
\end{enumerate}
\end{prop}

\begin{prop}[$\mathcal{CF}$-invariant projections]\label{prop:inv-proj}
Let $r\in \{1,2,\dots, d_A^2\}$ and let $P\in \Proj_r(AA^*)$ be $\mathcal{CF}$-invariant. Then all of the following hold.
\begin{enumerate}[label=(\alph*)]
\item If $\ket{\phi}\in \im (P)$, then also $\ket{\psi}\coloneqq (\swapf \ket{\phi})^*\in \im (P)$ and $\brak{\psi}=\brak{\phi}$.
\item There exists a $\mathcal{CF}$-invariant unit vector $\ket{\psi}\in \im (P)$.
\item There exist $\mathcal{CF}$-invariant projections $P_j\in \Proj_1(AA^*)$ for $j\in [r]$ such that $P=\sum_{j\in [r]}P_j$ and $P_jP_k=\delta_{j,k}P_{j}$ for all $j,k\in [r]$.
\end{enumerate}
\end{prop}
\begin{proof}[Proof of (a).]
Let $\ket{\phi}\in \im(P)$. 
Then, $P\ket{\phi}=\ket{\phi}$. 
Let $\ket{\psi}\coloneqq (\swapf \ket{\phi})^*$. 
Then,
\begin{align}
P\ket{\psi}
=\swapf  P^{\swapf }(\ket{\phi})^*
&=\swapf  (P^{\swapf *}\ket{\phi})^*
=\swapf  (P\ket{\phi})^*
=\swapf  (\ket{\phi})^* 
=\ket{\psi},
\end{align}
so $\ket{\psi}\in \im (P)$. Moreover,
$\brak{\psi}=(\brak{\phi})^*=\brak{\phi}$.
\end{proof}
\begin{proof}[Proof of (b).]
Let $\ket{\phi}\in \im (P)\setminus \{ 0\}$. Let $\ket{\psi_1}\coloneqq \ket{\phi}+(\swapf \ket{\phi})^*$. 

\emph{Case 1: $\ket{\psi_1}\neq 0$.} 
Let $c\coloneqq \sqrt{\brak{\psi_1}}$ and $\ket{\psi}\coloneqq \frac{1}{c}\ket{\psi_1}$. 
Then $\brak{\psi}=\frac{1}{c^2}\brak{\psi_1}=1$, 
so $\ket{\psi}$ is a unit vector. 
(a) implies that $\ket{\psi}\in \im (P)$ because 
$\ket{\psi}$ is a linear combination of $\ket{\phi}\in\im (P)$ and $(\swapf \ket{\phi})^*$. 
$\ket{\psi}$ is $\mathcal{CF}$-invariant because
\begin{align}
(\swapf \ket{\psi})^*
=\frac{1}{c}(\swapf \ket{\phi}+\ket{\phi}^*)^*
=\frac{1}{c}((\swapf \ket{\phi})^*+\ket{\phi})
=\ket{\psi} .
\end{align}

\emph{Case 2: $\ket{\psi_1}= 0$.} 
Then $(\swapf \ket{\phi})^*=-\ket{\phi}$.
Let $\ket{\psi_2}\coloneqq i( \ket{\phi}-(\swapf \ket{\phi})^*)=2i\ket{\phi}$, 
$c\coloneqq \sqrt{\brak{\psi_2}}$, and 
$\ket{\psi}\coloneqq \frac{1}{c}\ket{\psi_2}$. 
Then  
$\brak{\psi}=\frac{1}{c^2}\brak{\psi_2}=1$, 
so $\ket{\psi}$ is a unit vector. 
$\ket{\psi}\in \im (P)$ since $\ket{\psi}$ is proportional to $\ket{\phi}\in \im (P)$. 
$\ket{\psi}$ is $\mathcal{CF}$-invariant because
\begin{align}
(\swapf \ket{\psi})^*
=\frac{-i}{c}(\swapf \ket{\phi}-\ket{\phi}^*)^*
=\frac{-i}{c}((\swapf \ket{\phi})^*-\ket{\phi})
=\ket{\psi} .
\end{align}
\end{proof}
\begin{proof}[Proof of (c).]
\emph{Case 1: $r=1$.} 
The assertion holds trivially in this case, as one can simply define $P_0\coloneqq P$.

\emph{Case 2: $r>1$.} 
Using~(b), let $\ket{\psi}\in \im(P)$ be a $\mathcal{CF}$-invariant unit vector.
Let $P_{r-1}\coloneqq \proj{\psi}$. 
$P_{r-1}$ is $\mathcal{CF}$-invariant.
Let $Q\coloneqq P-P_{r-1}$. Then, 
$Q^{\swapf *}= P^{\swapf *} - P_{r-1}^{\swapf *}
=P - P_{r-1}
=Q$. 
Hence, $Q\in \Proj_{r-1}(AA^*)$ is $\mathcal{CF}$-invariant. 
The argument provided reduces proving the assertion for $P\in \Proj_{r}(AA^*)$ to proving the analogous assertion for $Q\in \Proj_{r-1}(AA^*)$. 
By recursively repeating the same argument, case~2 eventually reduces to case~1, thereby proving the assertion.
\end{proof}

\begin{prop}[$\mathcal{CF}$-invariant rank-one projections]\label{prop:inv-proj1}
Let $P\in \Proj_1(AA^*)$ be $\mathcal{CF}$-invariant. Then all of the following hold.
\begin{enumerate}[label=(\alph*)]
\item There exists a $\mathcal{CF}$-invariant unit vector $\ket{\psi}\in AA^*$ such that 
$P=\proj{\psi}$.
\item If $\ket{\psi}\in AA^*$ is a $\mathcal{CF}$-invariant unit vector such that 
$P=\proj{\psi}$, then also 
$\ket{\phi}\coloneqq -\ket{\psi}$ is a $\mathcal{CF}$-invariant unit vector such that 
$P=\proj{\phi}$.
\item If $\ket{\psi},\ket{\phi}\in AA^*$ are $\mathcal{CF}$-invariant unit vectors such that 
$P=\proj{\psi}=\proj{\phi}$, then either $\ket{\psi}=\ket{\phi}$ or $\ket{\psi}=-\ket{\phi}$.
\end{enumerate}
\end{prop}
\begin{proof}[Proof of (a).]
By Proposition~\ref{prop:inv-proj}~(b), there exists a $\mathcal{CF}$-invariant unit vector $\ket{\psi}\in \im (P)$. 
Since $P\in \Proj_1(AA^*)$, this implies that $P=\proj{\psi}$.
\end{proof}
\begin{proof}[Proof of (b).] 
Let $\ket{\psi}\in AA^*$ be a $\mathcal{CF}$-invariant unit vector such that 
$P=\proj{\psi}$. 
Let $\ket{\phi}\coloneqq -\ket{\psi}$. 
Since $\ket{\psi}$ is $\mathcal{CF}$-invariant, also $\ket{\phi}$ is $\mathcal{CF}$-invariant. 
Moreover, 
$\brak{\phi}=(-1)^2\brak{\psi}=\brak{\psi}=1$ 
and $\proj{\phi}=(-1)^2\proj{\psi}=\proj{\psi}=P$. 
\end{proof}
\begin{proof}[Proof of (c).]
Let $\ket{\psi},\ket{\phi}\in AA^*$ be $\mathcal{CF}$-invariant unit vectors such that $P=\proj{\psi}=\proj{\phi}$. 
The last equality implies that there exists $\theta \in [0,2\pi)$ such that
$\ket{\phi}=e^{i\theta}\ket{\psi}$. 
By $\mathcal{CF}$-invariance, we have 
$\ket{\phi}
=(\swapf \ket{\phi})^*
=e^{-i\theta}(\swapf \ket{\psi})^*
=e^{-i\theta}\ket{\psi}$. 
We can conclude that $e^{i\theta}=e^{-i\theta}$, which implies that either $\theta = 0$ or $\theta = \pi$. 
If $\theta = 0$, then $\ket{\phi}=\ket{\psi}$. 
If $\theta=\pi$, then $\ket{\phi}=-\ket{\psi}$.
\end{proof}

\begin{prop}[$\mathcal{CF}$-invariant unit vectors]\label{prop:canonical-invariant}
Let us define the following sets.
\begin{align}\label{eq:def-sunit}
\sunit
&\coloneqq \{\ket{\phi}\in AA^*:\brak{\phi}=1\}
\\
\sinvariant
&\coloneqq \{\ket{\psi}\in AA^*:\brak{\psi}=1,(\swapf \ket{\psi})^*=\ket{\psi}\}
\label{eq:def-sinvariant}
\\
\scanonical
&\coloneqq \{\ket{\chi}\in AA^*:\exists\, \sigma_A\in \mathcal{S}(A):\ket{\chi}=\sqrt{\sigma_A}\ket{\Omega}_{AA^*}\}
\label{eq:def-scanonical}
\end{align}
Then all of the following hold.
\begin{enumerate}[label=(\alph*)]
\item For any $\ket{\chi}\in AA^*:$ 
$\ket{\chi}\in \scanonical $ iff there exists an orthonormal basis $\{\ket{a_j}\}_{j\in [d_A]}$ for $A$ and $(p_j)_{j\in [d_A]}\in [0,1]^{\times d_A}$ such that $\sum_{j\in [d_A]}p_j=1$ and 
$\ket{\chi}=\sum_{j\in [d_A]}\sqrt{p_j}\ket{a_j}_{A}\otimes \ket{a_j}_{A^*}^*$.
\item For any $\ket{\psi}\in AA^*:$ 
$\ket{\psi}\in \sinvariant $ iff there exists an orthonormal basis $\{\ket{a_j}\}_{j\in [d_A]}$ for $A$ and $(s_j)_{j\in [d_A]}\in [-1,1]^{\times d_A}$ such that $\sum_{j\in [d_A]}s_j^2=1$ and 
$\ket{\psi}=\sum_{j\in [d_A]}s_j\ket{a_j}_{A}\otimes \ket{a_j}_{A^*}^*$.
\item $\sunit =\{\ket{\phi}\in AA^*:\exists\, \sigma_A\in \mathcal{S}(A),U_A\in \mathcal{U}(A):\ket{\phi}=\sqrt{\sigma_A} U_A\ket{\Omega}_{AA^*}\}$
\item $\sinvariant =\{\ket{\psi}\in AA^*:\exists\, \sigma_A\in \mathcal{S}(A),U_A\in \mathcal{U}(A):
\ket{\psi}=\sqrt{\sigma_A} U_A\ket{\Omega}_{AA^*},
[\sigma_A,U_A]=0,U_A^2=1_A\}$
\item $\scanonical \subseteq \sinvariant \subseteq \sunit $
\end{enumerate}
\end{prop}
\begin{proof}[Proof of~(a)]
Let $\ket{\chi}\in \scanonical$. 
By~\eqref{eq:def-scanonical}, there exists $\sigma_A\in \mathcal{S}(A)$ such that $\ket{\chi}=\sqrt{\sigma_A}\ket{\Omega}_{AA^*}$. 
The spectral theorem implies that there exists an orthonormal basis $\{\ket{a_j}\}_{j\in [d_A]}$ for $A$ and $(p_j)_{j\in [d_A]}\in [0,1]^{\times d_A}$ such that $\sigma_A=\sum_{j\in [d_A]}p_j\proj{a_j}$. 
Substituting this expression for $\sigma_A$ yields 
$\ket{\chi}=\sum_{j\in [d_A]}\sqrt{p_j}\ket{a_j}_{A}\otimes \ket{a_j}_{A^*}^*$. 
Moreover, $\sum_{j\in [d_A]}p_j=\tr[\sigma_A]=1$.

Conversely, let $\ket{\chi}\in AA^*$ be such that there exist an orthonormal basis $\{\ket{a_j}\}_{j\in [d_A]}$ for $A$ and $(p_j)_{j\in [d_A]}\in [0,1]^{\times d_A}$ such that 
$\sum_{j\in [d_A]}p_j=1$ 
and $\ket{\chi}=\sum_{j\in [d_A]}\sqrt{p_j}\ket{a_j}_{A}\otimes \ket{a_j}_{A^*}^*$. 
Let $\sigma_A\coloneqq \sum_{j\in [d_A]}p_j\proj{a_j}$. 
Then, $\sigma_A\in \mathcal{S}(A)$ and $\ket{\chi}=\sqrt{\sigma_A}\ket{\Omega}_{AA^*}$.
\end{proof}
\begin{proof}[Proof of~(b)]
Let $\ket{\psi}\in \sinvariant $.
By~\eqref{eq:def-sinvariant}, $\ket{\psi}$ is a unit vector such that $(\swapf \ket{\psi})^*=\ket{\psi}$.  
Let $C_{j,k}\coloneqq \bra{j}_A\otimes\bra{k}_{A^*}\ket{\psi}$ for all $j,k\in [d_A]$. 
Then, for all $j,k\in [d_A]$
\begin{align}
C_{j,k}
=\bra{j}_A\otimes\bra{k}_{A^*}(\swapf \ket{\psi})^*
=(\bra{j}_A\otimes\bra{k}_{A^*}\swapf \ket{\psi})^*
=(\bra{k}_A\otimes\bra{j}_{A^*}\ket{\psi})^*
=C_{k,j}^* .
\label{eq:c-ij}
\end{align}
Let us define the $(d_A\times d_A)$-matrix $C\coloneqq (C_{j,k})_{j\in [d_A],k\in [d_A]}$.
\eqref{eq:c-ij} implies that $C^\dagger =C$. 
Hence, there exist $(d_A\times d_A)$-matrices $U\equiv (U_{j,k})_{j\in [d_A],k\in [d_A]} , D\equiv (D_{j,k})_{j\in [d_A],k\in [d_A]}$ such that $C=UDU^\dagger$, $U$ is unitary, and $D_{j,k}=\delta_{j,k}s_j$ for $s_j\in \mathbb{R}$.  
Let $\ket{a_k}_A\coloneqq \sum_{j\in [d_A]} U_{j,k}\ket{j}_A$ for all $k\in [d_A]$. 
$\{ \ket{a_j} \}_{j\in [d_A]}$ is an orthonormal basis for $A$ because $U$ is unitary. 
Moreover, we have 
\begin{align}
\ket{\psi}
&=\sum_{j,k\in [d_A]}C_{j,k}\ket{j}_A\otimes\ket{k}_{A^*}
=\sum_{j,k,l\in [d_A]}U_{j,l}s_l U_{k,l}^*\ket{j}_A\otimes\ket{k}_{A^*}\\
&=\sum_{l\in [d_A]}s_l \left(\sum_{j\in [d_A]}U_{j,l}\ket{j}_A\right) \otimes \left(\sum_{k\in [d_A]}U_{k,l}\ket{k}_{A^*}\right)^*
=\sum_{l\in [d_A]}s_l\ket{a_l}_{A}\otimes \ket{a_l}_{A^*}^*.
\end{align}
Since $\ket{\psi}$ is a unit vector, 
$1=\brak{\psi}=\sum_{j,k\in [d_A]}\lvert C_{j,k}\rvert^2=\sum_{l\in [d_A]}s_l^2$. 
This implies that $s_j\in [-1,1]$ for all $j\in [d_A]$.

Conversely, let $\ket{\psi}\in AA^*$ be such that there exist an orthonormal basis $\{\ket{a_j}\}_{j\in [d_A]}$ for $A$ and $(s_j)_{j\in [d_A]}\in [-1,1]^{\times d_A}$ such that $\sum_{j\in [d_A]}s_j^2=1$ and 
$\ket{\psi}=\sum_{j\in [d_A]}s_j\ket{a_j}_{A}\otimes \ket{a_j}_{A^*}^*$. 
Then, $\brak{\psi}=\sum_{j\in [d_A]}s_j^2=1$ and 
\begin{align}
(\swapf \ket{\psi})^*
=\sum_{j\in [d_A]}s_j(\swapf \ket{a_j}_{A}\otimes \ket{a_j}_{A^*}^*)^*
=\sum_{j\in [d_A]}s_j \ket{a_j}_{A}\otimes \ket{a_j}_{A^*}^*
=\ket{\psi} .
\end{align}
\end{proof}
\begin{proof}[Proof of~(c)]
Let $\ket{\phi}\in \sunit $. 
Let $\sigma_A\coloneqq \tr_{A^*}[\proj{\phi}]\in \mathcal{S}(A)$. 
On the one hand, $\ket{\phi}$ is a purification of $\sigma_A$.
On the other hand, also $\sqrt{\sigma_A}\ket{\Omega}_{AA^*}$ is a purification of $\sigma_A$.
Hence, there exists $V_{A^*}\in \mathcal{U}(A^*)$ such that $\ket{\phi}=\sqrt{\sigma_A}\otimes V_{A^*}\ket{\Omega}_{AA^*}$. 
Hence, $\ket{\phi}=\sqrt{\sigma_A}U_A\ket{\Omega}_{AA^*}$ for $U_A\coloneqq (V_A^\dagger)^*\in \mathcal{U}(A)$.

Conversely, let $\sigma_A\in \mathcal{S}(A),U_A\in \mathcal{U}(A)$, and $\ket{\phi}\coloneqq\sqrt{\sigma_A} U_A\ket{\Omega}_{AA^*}$. 
Then
\begin{equation}
\brak{\phi}=\bra{\Omega}_{AA^*}U_A^\dagger\sqrt{\sigma_A} \sqrt{\sigma_A} U_A\ket{\Omega}_{AA^*}
=\tr[U_A^\dagger \sigma_A U_A]
=\tr[\sigma_A]=1 .
\end{equation}
\end{proof}
\begin{proof}[Proof of~(d)] 
Let $\ket{\psi}\in \sinvariant $. 
By~(b), there exist an orthonormal basis $\{\ket{a_j}\}_{j\in [d_A]}$ for $A$ and $(s_j)_{j\in [d_A]}\in [-1,1]^{\times d_A}$ such that $\sum_{j\in [d_A]}s_j^2=1$ and 
$\ket{\psi}=\sum_{j\in [d_A]}s_j \ket{a_j}_{A}\otimes \ket{a_j}_{A^*}^*$. 
Let $\sigma_A\coloneqq \tr_{A^*}[\proj{\psi}]=\sum_{j\in [d_A]}s_j^2 \proj{a_j}$.
Let $u_j\coloneqq +1$ if $s_j\geq 0$ and $u_j\coloneqq -1$ if $s_j<0$ for all $j\in [d_A]$,
and let $U_A\coloneqq \sum_{j\in [d_A]}u_j \proj{a_j}_{A}\in \mathcal{U}(A)$. 
Then, $[\sigma_A,U_A]=0$ and $U_A^2=1_A$. Moreover, 
\begin{align}
\sqrt{\sigma_A} U_A\ket{\Omega}_{AA^*}
=\sum_{j\in [d_A]}\lvert s_j\rvert\cdot u_j \ket{a_j}_{A}\otimes \ket{a_j}_{A^*}^*
=\sum_{j\in [d_A]}s_j\ket{a_j}_{A}\otimes \ket{a_j}_{A^*}^*
=\ket{\psi}.
\end{align}

Conversely, let $\sigma_A\in \mathcal{S}(A),U_A\in \mathcal{U}(A)$ be such that $[\sigma_A,U_A]=0,U_A^2=1_A$, and let $\ket{\psi}\coloneqq \sqrt{\sigma_A} U_A\ket{\Omega}_{AA^*}$. 
Since $U_A$ is unitary and satisfies $U_A^2=1_A$, it follows that $U_A^\dagger=U_A$.  
Moreover, 
\begin{align}
(\swapf \ket{\psi})^*
&=\sum_{j,k,l\in [d_A]}\ket{j}_A\otimes\ket{k}_{A^*} (\bra{k}_A\otimes\bra{j}_{A^*}\sqrt{\sigma_A}U_A\ket{l}_{A}\otimes\ket{l}_{A^*})^*
\\
&=\sum_{j,k\in [d_A]}\ket{j}_A\otimes\ket{k}_{A^*} (\bra{k}\sqrt{\sigma_A}U_A\ket{j}_A)^*
\\
&=\sum_{j,k\in [d_A]}\ket{j}_A\otimes\ket{k}_{A^*} \bra{j}(\sqrt{\sigma_A}U_A)^\dagger\ket{k}_A
\\
&=(\sqrt{\sigma_A} U_A)^\dagger\ket{\Omega}_{AA^*}
=U_A^\dagger\sqrt{\sigma_A}^\dagger \ket{\Omega}_{AA^*}
\\
&= U_A\sqrt{\sigma_A}\ket{\Omega}_{AA^*}
=\sqrt{\sigma_A}U_A\ket{\Omega}_{AA^*}
=\ket{\psi},
\end{align}
and $\brak{\psi}=\tr[U_A^\dagger \sigma_A U_A]
=\tr[\sigma_A]=1$. 
\end{proof}
\begin{proof}[Proof of (e)] 
The inclusion $\scanonical \subseteq \sinvariant $ follows from~\eqref{eq:def-scanonical} and~(d) by choosing $U_A\coloneqq 1_A$. 
The inclusion $\sinvariant \subseteq \sunit $ follows from the definition of these sets in~\eqref{eq:def-sunit} and~\eqref{eq:def-sinvariant}.
\end{proof}

\begin{prop}[Spectral decomposition for $\mathcal{CF}$-invariant self-adjoint operators]\label{prop:inv-spectral}
Let $X\in \mathcal{L}(AA^*)$ be self-adjoint and $\mathcal{CF}$-invariant.
Then there exist an orthonormal basis $\{\ket{e_j}\}_{j\in [d_A^2]}$ for $AA^*$ whose elements are $\mathcal{CF}$-invariant and $(x_j)_{j\in [d_A]}\in \mathbb{R}^{\times d_A}$ such that 
$X=\sum_{j\in [d_A^2]}x_j \proj{e_j}$.
\end{prop}
\begin{proof}
Let $n$ be the number of different eigenvalues of $X$, i.e., $n\coloneqq \lvert \spec(X)\rvert$. 
Let the eigenvalues of $X$ be denoted by $\tilde{x}_j\in \mathbb{R}$ for $j\in [n]$. 
For each $j\in [n]$, let $P_j\in \Proj(AA^*)$ be the orthogonal projection onto the eigenspace associated with the eigenvalue $\tilde{x}_j$. Then, 
\begin{equation}\label{eq:x-xi-pi}
X=\sum_{j\in [n]}\tilde{x}_jP_j\,,
\end{equation}
$\sum_{j\in [n]}P_j=1_{AA^*}$, and $P_jP_k=\delta_{j,k}P_j$ for all $j,k\in [n]$.
Since $X$ is $\mathcal{CF}$-invariant, 
\begin{equation}\label{eq:x-xi-qi}
X=X^{\swapf *}=\sum_{j\in [n]}\tilde{x}_jP_j^{\swapf *}=\sum_{j\in [n]}\tilde{x}_jQ_j,
\end{equation}
where we defined $Q_j\coloneqq P_j^{\swapf *}$ for all $j\in [n]$. 
Then, $Q_j^\dagger =Q_j$ for all $j\in [n]$, and
\begin{equation}
Q_jQ_k=\swapf P_j^*P_k^*\swapf 
=\swapf \delta_{j,k}P_j^*\swapf 
=\delta_{j,k}P_j^{\swapf *}
=\delta_{j,k}Q_j
\end{equation}
for all $j,k\in [n]$. 
Hence, $Q_j\in \Proj(AA^*)$ for all $j\in [n]$.
Moreover, $\sum_{j\in [n]}Q_j=(\sum_{j\in [n]}P_j)^{\swapf *}=(1_{AA^*})^{\swapf *}=1_{AA^*}$.
Hence, the expression on the far right in~\eqref{eq:x-xi-qi} is a spectral decomposition of $X$.
Since spectral decompositions into eigenspaces are unique, it follows from~\eqref{eq:x-xi-pi} and~\eqref{eq:x-xi-qi} that $P_j=Q_j$ for all $j\in [n]$. 
Since $Q_j= P_j^{\swapf *}$, this implies that $P_j$ is $\mathcal{CF}$-invariant for all $j\in [n]$.

Let $r_j\coloneqq \rank (P_j)$ for all $j\in [n]$.
By Proposition~\ref{prop:inv-proj}~(c) and Proposition~\ref{prop:inv-proj1}~(a), it follows that for each $j\in [n]$, there exists a collection of $\mathcal{CF}$-invariant unit vectors $\ket{b_{j,k}}\in AA^*$ for $k\in [r_j]$ such that $P_j=\sum_{k\in [r_j]}\proj{b_{j,k}}$ and $\braket{b_{j,k}}{b_{j,l}}=\delta_{k,l}$ for all $k,l\in [r_j]$. 
Since the projections $P_j$ are mutually orthogonal and span the full Hilbert space, it follows that $\cup_{j\in [n]}\{\ket{b_{j,k}}\}_{k\in [r_j]}$ is an orthonormal basis for $AA^*$. 
By~\eqref{eq:x-xi-pi}, 
$X=\sum_{j\in [n]}\sum_{k\in [r_j]}\tilde{x}_j\proj{b_{j,k}}$.
\end{proof}

\begin{prop}[Optimization problems for $\mathcal{CF}$-invariant self-adjoint operators]\label{prop:inv-self-adjoint}
Let $X\in \mathcal{L}(AA^*)$ be self-adjoint and $\mathcal{CF}$-invariant. 
Then all of the following hold.
\begin{enumerate}[label=(\alph*)]
\item \emph{Optimization for general rank:} 
Let $r\in \{1,2,\dots, d_A^2\}$. Then
\begin{align}\label{eq:max-xp-xq}
\max_{P\in \Proj_r(AA^*)} \tr[XP]
&=\max_{\substack{Q\in \Proj_r(AA^*):\\ Q^{\swapf *}=Q}} \tr[XQ] .
\end{align}
\item \emph{Optimization for rank-one:} 
\begin{align}\label{eq:max-phi-psi}
\max_{\ket{\phi}\in AA^*:\brak{\phi}=1} \tr[X\proj{\phi}]
&=\max_{\substack{\ket{\psi}\in AA^*:\brak{\psi}=1,\\ (\swapf \ket{\psi})^*=\ket{\psi} }} \tr[X\proj{\psi}] .
\end{align}
\item \emph{Characterization of optimizers for~\eqref{eq:max-phi-psi}:}
Let $\ket{\phi}$ be an optimizer for the left-hand side of~\eqref{eq:max-phi-psi}.
Let 
\begin{equation}\label{eq:def-saa}
S\coloneqq \frac{1}{2}\proj{\phi}+\frac{1}{2}(\proj{\phi})^{\swapf *} .
\end{equation}
Then $S\geq 0$.

If the largest eigenvalue of $S$ is non-degenerate:
Let $\ket{e}\in AA^*$ be a unit eigenvector of $S$ associated with the largest eigenvalue of $S$.
Then there exists $\theta\in [0,2\pi)$ such that 
$\ket{\psi}\coloneqq e^{i\theta}\ket{e}$ 
is an optimizer for the right-hand side of~\eqref{eq:max-phi-psi}.

Else: Let $\ket{\psi}\coloneqq \frac{1}{\sqrt{2}}(\ket{\phi}+(\swapf \ket{\phi})^*)$.
Then $\ket{\psi}$
is an optimizer for the right-hand side of~\eqref{eq:max-phi-psi}.
\end{enumerate}
\end{prop}
\begin{proof}[Proof of~(a)]
Since the feasible set on the right-hand side of~\eqref{eq:max-xp-xq} is a subset of the feasible set on the left-hand side of~\eqref{eq:max-xp-xq}, it is trivially true that
\begin{equation}\label{eq:proof-pq1}
\max_{P\in \Proj_r(AA^*)} \tr[XP]
\geq \sup_{\substack{Q\in \Proj_r(AA^*):\\Q^{\swapf *}=Q}} \tr[XQ] .
\end{equation}

Next, we will prove the opposite inequality. 
Let $P\in \Proj_r(AA^*)$. Then
\begin{equation}\label{eq:x-cc}
\tr[XP]
=\tr[XP]^*
=\tr[X^{\swapf }P^{\swapf }]^*
=\tr[X^{\swapf *}P^{\swapf *}]
=\tr[XP^{\swapf *}] .
\end{equation}
Let $S\coloneqq \frac{1}{2}P+\frac{1}{2}P^{\swapf *}$.
Then, $S$ is $\mathcal{CF}$-invariant and self-adjoint. 
By Proposition~\ref{prop:inv-spectral}, there exist an orthonormal basis $\{\ket{e_j}\}_{j\in [d_A^2]}$ for $AA^*$ whose elements are $\mathcal{CF}$-invariant and $(s_j)_{j\in [d_A^2]}\in \mathbb{R}^{\times d_A^2}$ such that 
$S=\sum_{j\in [d_A^2]}s_j\proj{e_j}$. 
We have
\begin{align}
\max(\spec(S))&\leq \frac{1}{2}\max (\spec(P))+\frac{1}{2}\max (\spec(P^{\swapf *}))
=\max(\spec (P))=1,
\end{align}
and $S\geq 0$ because $P\geq 0$. 
Hence, $s_j\in [0,1]$ for all $j\in [d_A^2]$. 
Since $\tr[S]=\tr[P]=r$, it follows that $\sum_{j\in [d_A^2]}s_j=r$. 

Without loss of generality, suppose  
$\tr[X\proj{e_0}]\geq \tr[X\proj{e_1}]\geq \dots \geq \tr[X\proj{e_{d_A^2-1}}]$. 
(This can always be achieved by jointly relabeling the $s_j$'s and $\ket{e_j}$'s.) 
Let $Q\coloneqq \sum_{j\in [r]}\proj{e_j}$. 
Then
\begin{align}
\tr[XP]
\stackrel{\eqref{eq:x-cc}}{=}\tr[XS]
&=\sum_{j\in [d_A^2]}s_j\tr[X\proj{e_j}]
\leq \sum_{j\in [r]}\tr[X\proj{e_j}]
=\tr[XQ]. 
\label{eq:xp-ei-xq}
\end{align}
For the inequality, we used that $s_j\in [0,1]$ for all $j\in [d_A^2]$ and $\sum_{j\in [d_A^2]}s_j=r$. 
The definition of $Q$ implies that $Q\in \Proj_r(AA^*)$ and that $Q^{\swapf *}=Q$.
We can conclude from~\eqref{eq:xp-ei-xq} that
\begin{equation}\label{eq:proof-pq2}
\max_{P\in \Proj_r(AA^*)} \tr[XP]
\leq \sup_{\substack{Q\in \Proj_r(AA^*):\\Q^{\swapf *}=Q}} \tr[XQ] .
\end{equation}

The assertion now follows from~\eqref{eq:proof-pq1} and~\eqref{eq:proof-pq2}. 
\end{proof}
\begin{proof}[Proof of~(b)]
The assertion follows from~(a) for $r=1$ due to Proposition~\ref{prop:inv-proj1}~(a).
\end{proof}
\begin{proof}[Proof of~(c)]
By the definition of $S$ in~\eqref{eq:def-saa}, we have $\rank(S)=1$ or $\rank(S)=2$.

\emph{Case 1: $\rank(S)=1$.} 
By the proof of~(a) for $r=1$ and $P\coloneqq \proj{\phi}$, there exists an orthonormal basis $\{\ket{e_j}\}_{j\in [d_A^2]}$ for $AA^*$ whose elements are $\mathcal{CF}$-invariant such that 
$S=\proj{e_0}$. 
Since $\ket{\phi}$ is an optimizer for the left-hand side of~\eqref{eq:max-phi-psi}, the inequality in~\eqref{eq:xp-ei-xq} must be saturated.
This implies that $\ket{e_0}$ 
is an optimizer for the right-hand side of~\eqref{eq:max-phi-psi}.

For case 1, the largest eigenvalue of $S$ is non-degenerate. 
Let $\ket{e}$ be a unit eigenvector of $S$ associated with the largest eigenvalue of $S$.
Then, $\proj{e}=S=\proj{e_0}$. 
Hence, there exists $\theta\in [0,2\pi)$ such that $\ket{e}=e^{-i\theta}\ket{e_0}$.

\emph{Case 2: $\rank(S)=2$ and the largest eigenvalue of $S$ is non-degenerate.} 
By the proof of~(a) for $r=1$ and $P\coloneqq \proj{\phi}$, there exists an orthonormal basis $\{\ket{e_j}\}_{j\in [d_A^2]}$ for $AA^*$ whose elements are $\mathcal{CF}$-invariant such that 
$S=\sum_{i=0}^1s_j\proj{e_j}$ for certain $s_j\in (0,1)$, 
$\sum_{j=0}^1s_j=1$ and $\tr[X\proj{e_0}]\geq \tr[X\proj{e_1}]$.
Since $\ket{\phi}$ is an optimizer for the left-hand side of~\eqref{eq:max-phi-psi}, the inequality in~\eqref{eq:xp-ei-xq} must be saturated.
This implies that both $\ket{e_0}$ and $\ket{e_1}$ 
are optimizers for the right-hand side of~\eqref{eq:max-phi-psi}.

Since the largest eigenvalue of $S$ is non-degenerate, we have $s_0\neq s_1$.
Let $\ket{e}$ be a unit eigenvector of $S$ associated with the largest eigenvalue of $S$.
Then, either $\proj{e}=\proj{e_0}$ or $\proj{e}=\proj{e_1}$.
Hence, there exists $\theta\in [0,2\pi)$ such that $\ket{e}=e^{-i\theta}\ket{e_0}$ or $\ket{e}=e^{-i\theta}\ket{e_1}$.

\emph{Case 3: $\rank(S)=2$ and the largest eigenvalue of $S$ is degenerate.} 
Since $S$ is positive semidefinite and $\tr[S]=1$, 
it follows that $\spec(S)\subseteq \{0,\frac{1}{2}\}$.
By~\eqref{eq:def-saa}, $\ket{\phi}$ and $(\swapf \ket{\phi})^*$ are orthogonal unit eigenvectors of $S$ associated with the eigenvalue $\frac{1}{2}$.
Let $\ket{\psi}\coloneqq \frac{1}{\sqrt{2}}(\ket{\phi}+(\swapf \ket{\phi})^*)$ and 
$\ket{\psi^\perp}\coloneqq \frac{i}{\sqrt{2}}(\ket{\phi}-(\swapf \ket{\phi})^*)$.
Then, $\ket{\psi}$ and $\ket{\psi^\perp }$ are orthogonal unit eigenvectors of $S$ associated with the eigenvalue $\frac{1}{2}$, and they are $\mathcal{CF}$-invariant. 
Hence, 
$S=\frac{1}{2}\proj{\psi}+\frac{1}{2}\proj{\psi^\perp }
= \sum_{j=0}^1s_j\proj{e_j}$ 
where we defined $s_0\coloneqq \frac{1}{2}$, $s_1\coloneqq \frac{1}{2}$, $\ket{e_0}\coloneqq \ket{\psi}$, and $\ket{e_1}\coloneqq \ket{\psi^\perp }$.
This is a decomposition of the form used in the proof of~(a) for $r=1$ and $P\coloneqq \proj{\phi}$. 
Since $\ket{\phi}$ is an optimizer for the left-hand side of~\eqref{eq:max-phi-psi}, the inequality in~\eqref{eq:xp-ei-xq} must be saturated. 
This implies that both $\ket{\psi}$ and $\ket{\psi^\perp}$ 
are optimizers for the right-hand side of~\eqref{eq:max-phi-psi}.
\end{proof}

\begin{prop}[$\mathcal{CF}$-invariance is preserved under partial trace]\label{prop:inv-global}
All of the following hold.
\begin{enumerate}[label=(\alph*)]
\item Let $X_{ABA^*B^*}\in \mathcal{L}(ABA^*B^*)$ be such that
$(X_{ABA^*B^*}^{\swapf_{AA^*} \swapf_{BB^*} })^*=X_{ABA^*B^*}$. 
Let $Y_{AA^*}\coloneqq \tr_{BB^*}[X_{ABA^*B^*}]$.
Then, $Y_{AA^*}$ is $\mathcal{CF}$-invariant.
\item \sloppy 
Let $\rho_{AB}\in \mathcal{S}(AB)$, 
let $\ket{\hat{\rho}}_{ABA^*B^*}$ be its canonical purification, 
and let $\hat{\rho}_{AA^*}\coloneqq \tr_{BB^*}[\proj{\hat{\rho}}_{ABA^*B^*}]$. 
Then, $\hat{\rho}_{AA^*}$ is $\mathcal{CF}$-invariant.
\end{enumerate}
\end{prop}
\begin{proof}[Proof of~(a)]
\begin{align}
Y_{AA^*}^{\swapf *}
&=\swapf_{AA^*} (\tr_{BB^*}[X_{ABA^*B^*}])^* \swapf _{AA^*}
=\swapf_{AA^*} (\tr_{BB^*}[(X_{ABA^*B^*}^{\swapf_{AA^*}\swapf_{BB^*}})^*])^* \swapf_{AA^*}
\\
&=\swapf_{AA^*} \tr_{BB^*}[X_{ABA^*B^*}^{\swapf_{AA^*} \swap{B}{B^*}}] \swapf_{AA^*}
=\tr_{BB^*}[X_{ABA^*B^*}^{\swap{B}{B^*}}]
=Y_{AA^*}
\end{align}
\end{proof}
\begin{proof}[Proof of~(b)]
The canonical purification is given by 
$\ket{\hat{\rho}}_{ABA^*B^*}=\sqrt{\rho_{AB}}\ket{\Omega}_{AA^*}\otimes \ket{\Omega}_{BB^*}$. 
Let $X_{ABA^*B^*}\coloneqq \proj{\hat{\rho}}_{ABA^*B^*}\in \mathcal{L}(ABA^*B^*)$. 
Then, $(X_{ABA^*B^*}^{\swapf_{AA^*} \swapf_{BB^*} })^*=X_{ABA^*B^*}$ due to Proposition~\ref{prop:canonical-invariant}~(e). 
Since $\hat{\rho}_{AA^*}=\tr_{BB^*}[X_{ABA^*B^*}]$, the assertion now follows from~(a).
\end{proof}

\section{Proofs, remarks, and examples}

\subsection{Proof of Proposition~\ref{prop:i12-omega}}\label{app:proof-omega}
\begin{proof}  
\emph{Case 1: $\rho_A\not\perp \sigma_A$.} 
Then 
\begin{align}
\exp(-I_{1/2}^{\downarrow}(\rho_{AB}\| \sigma_A ))
&= \tr[( \tr_A[\sqrt{\rho_{AB}}\sqrt{\sigma_A}] )^2 ]
\label{eq:i12-rs1}\\
&= \tr[ \tr_A[\sqrt{\rho_{AB}}\sqrt{\sigma_A}] \tr_A[\sqrt{\sigma_A}\sqrt{\rho_{AB}}] ]\\
&=\tr[\bra{\Omega}_{AA^*}\sqrt{\rho_{AB}}\sqrt{\sigma_A}\proj{\Omega}_{AA^*}\sqrt{\sigma_A}\sqrt{\rho_{AB}}\ket{\Omega}_{AA^*}]\\
&=\tr[ \sqrt{\sigma_A}\proj{\Omega}_{AA^*}\sqrt{\sigma_A} \proj{\hat{\rho}}_{ABA^*B^*} ]\\
&=\tr[ \sqrt{\sigma_A}\proj{\Omega}_{AA^*}\sqrt{\sigma_A} \hat{\rho}_{AA^*} ]\\
&=\bra{\Omega}_{AA^*} \sqrt{\sigma_A}\hat{\rho}_{AA^*}\sqrt{\sigma_A}\ket{\Omega}_{AA^*}.
\end{align}
\eqref{eq:i12-rs1} follows from~\eqref{eq:i-gen-explicit}. 

\emph{Case 2: $\rho_A\perp \sigma_A$.} 
Then $I_{1/2}^{\downarrow}(\rho_{AB}\| \sigma_A )=\infty$.
\end{proof}

\subsection{Proof of Theorem~\ref{thm:srinf-i12}}
\label{app:proof-thm}
\begin{proof}
Let $\ket{\hat{\rho}}_{ABA^*B^*}$ be the canonical purification of $\rho_{AB}$ 
and let $\hat{\rho}_{AA^*}\coloneqq \tr_{BB^*}[\proj{\hat{\rho}}_{ABA^*B^*}]$. 
Let $\sigma_A\in \mathcal{S}(A)$ and $\ket{\chi}_{AA^*}\coloneqq\sqrt{\sigma_A}\ket{\Omega}_{AA^*}$
By Proposition~\ref{prop:i12-omega}, 
\begin{align}\label{eq:thm-proof0}
\exp(-I_{1/2}^{\downarrow}(\rho_{AB}\| \sigma_A ))
&=\tr[\hat{\rho}_{AA^*}\sqrt{\sigma_A}\proj{\Omega}_{AA^*}\sqrt{\sigma_A}]
=\tr[\hat{\rho}_{AA^*}\proj{\chi}]
\end{align}
if $\rho_A\not\perp\sigma_A$, and $\exp(-I_{1/2}^{\downarrow}(\rho_{AB}\| \sigma_A ))=\exp(-\infty)=0$ else.
Therefore,
\begin{align}
\exp(-I_{1/2}^{\downarrow\downarrow}(A:B)_\rho)
&=\max_{\sigma_A\in \mathcal{S}(A)}\exp(-I_{1/2}^{\downarrow}(\rho_{AB}\| \sigma_A ))
=\max_{\substack{\ket{\chi}\in AA^*:\\ \exists \sigma_A\in \mathcal{S}(A): \\ 
\ket{\chi}=\sqrt{\sigma_A}\ket{\Omega}_{AA^*} }} \tr[\hat{\rho}_{AA^*}\proj{\chi}].
\label{eq:thm-proof1}
\end{align}
For the min-reflected entropy holds
\begin{align}
\exp(-S_{R}^{(\infty)}(A:B)_\rho)
&=\max_{\substack{\ket{\phi}\in AA^*:\\ \brak{\phi}=1}}\tr[\hat{\rho}_{AA^*}\proj{\phi}]
=\max_{\substack{\ket{\psi}\in AA^*:\brak{\psi}=1,\\ (\swapf \ket{\psi})^*=\ket{\psi}  }}
\tr[\hat{\rho}_{AA^*}\proj{\psi}],
\label{eq:thm-proof2}
\end{align}
where we used the definition of $S_R^{(\infty)}(A:B)_\rho$ and~\eqref{eq:h-inf} for the first equality, 
and Proposition~\ref{prop:inv-self-adjoint}~(b) and Proposition~\ref{prop:inv-global}~(b) for the second equality. 
According to Proposition~\ref{prop:canonical-invariant}~(e), the feasible set on the far right of~\eqref{eq:thm-proof1} is a subset of the feasible set on the far right of~\eqref{eq:thm-proof2}. 
Therefore, we can conclude that 
$\exp(-I_{1/2}^{\downarrow\downarrow}(A:B)_\rho)\leq \exp(-S_{R}^{(\infty)}(A:B)_\rho)$. 
It remains to prove the opposite inequality. 

Let $\ket{\psi}\in AA^*$ be a $\mathcal{CF}$-invariant unit vector.
By Proposition~\ref{prop:canonical-invariant}~(b), 
there exist an orthonormal basis $\{\ket{a_j}\}_{j\in [d_A]}$ for $A$ 
and $(s_j)_{j\in [d_A]}\in [-1,1]^{\times d_A}$ 
such that $\sum_{j\in [d_A]}s_j^2=1$ and 
$\ket{\psi}=\sum_{j\in [d_A]}s_j\ket{a_j}_{A}\otimes \ket{a_j}_{A^*}^*$.
Let $\ket{\chi}\coloneqq \sum_{j\in [d_A]}\sqrt{p_j}\ket{a_j}_{A}\otimes \ket{a_j}_{A^*}^*$ where $p_j\coloneqq s_j^2$ for all $j\in [d_A]$. 
Then,
\begin{align}
\tr[\hat{\rho}_{AA^*}\proj{\chi}]
&=\sum_{j,k\in [d_A]}\sqrt{p_jp_k} \tr[\ketbra{a_j}{a_k}_{A}\otimes (\ketbra{a_j}{a_k}_{A^*})^*\sqrt{\rho_{AB}}\proj{\Omega}_{AA^*}\sqrt{\rho_{AB}} ]
\label{eq:s5}\\
&= \sum_{j,k\in [d_A]}\lvert s_j\rvert \lvert s_k\rvert \tr[\bra{a_j}_A\sqrt{\rho_{AB}}\ket{a_j}_A \bra{a_k}_A\sqrt{\rho_{AB}}\ket{a_k}_A ]
\\
&= \sum_{j,k\in [d_A]} \lvert s_js_k\tr[\bra{a_j}_A\sqrt{\rho_{AB}}\ket{a_j}_A \bra{a_k}_A\sqrt{\rho_{AB}}\ket{a_k}_A ] \rvert
\\
&\geq \sum_{j,k\in [d_A]}s_js_k\tr[\bra{a_j}_A\sqrt{\rho_{AB}}\ket{a_j}_A \bra{a_k}_A\sqrt{\rho_{AB}}\ket{a_k}_A ]
\label{eq:s56}\\
&=\sum_{j,k\in [d_A]}s_js_k \tr[\ketbra{a_j}{a_k}_{A}\otimes (\ketbra{a_j}{a_k}_{A^*})^*\sqrt{\rho_{AB}}\proj{\Omega}_{AA^*}\sqrt{\rho_{AB}} ]
\\
&=\tr[\hat{\rho}_{AA^*}\proj{\psi}] .
\label{eq:s6}
\end{align}
Let $\sigma_A\coloneqq \tr_{A^*}[\proj{\chi}]=\sum_{j\in [d_A]}p_j\proj{a_j}\in \mathcal{S}(A)$. 
Then, $\ket{\chi}_{AA^*}=\sqrt{\sigma_A}\ket{\Omega}_{AA^*}$.
Therefore, 
\begin{equation}
\max_{\substack{\ket{\chi}\in AA^*:\\ \exists \sigma_A\in \mathcal{S}(A): \\ \ket{\chi}=\sqrt{\sigma_A}\ket{\Omega}_{AA^*} }} \tr[\hat{\rho}_{AA^*}\proj{\chi}]
\geq
\max_{\substack{\ket{\psi}\in AA^*:\brak{\psi}=1,\\ (\swapf \ket{\psi})^*=\ket{\psi}  }}
\tr[\hat{\rho}_{AA^*}\proj{\psi}] .
\end{equation}
By~\eqref{eq:thm-proof1} and~\eqref{eq:thm-proof2}, we can conclude that 
$\exp(-I_{1/2}^{\downarrow\downarrow}(A:B)_\rho)\geq \exp(-S_{R}^{(\infty)}(A:B)_\rho)$.
\end{proof}

\subsection{Remarks on the proof method for Theorem~\ref{thm:srinf-i12}}
\label{app:thm-remarks}
The proof of Theorem~\ref{thm:srinf-i12} uses the $\mathcal{CF}$-invariance of $\hat{\rho}_{AA^*}$ (Proposition~\ref{prop:inv-global}~(b)) as a central ingredient, see~\eqref{eq:thm-proof2}. 
As argued in~\eqref{eq:thm-proof2}, this property of $\hat{\rho}_{AA^*}$ allows us to restrict an optimization over all unit vectors to an optimization over $\mathcal{CF}$-invariant unit vectors only (Proposition~\ref{prop:inv-self-adjoint}~(b)). 
By means of a triangle inequality, see~\eqref{eq:s56}, we have then shown that the feasible set can be restricted even more to the set of unit vectors that are canonical purifications (Proposition~\ref{prop:canonical-invariant}~(e)), see~\eqref{eq:thm-proof1}. 
To illustrate the usefulness of this proof technique, we consider the following corollary of Theorem~\ref{thm:srinf-i12}. 
\begin{cor}[Min-reflected entropy in terms of $\rho_{AB}$]\label{cor:max-u}
Let $\rho_{AB}\in \mathcal{S}(AB)$. Then 
\begin{align}
\exp(-S_R^{(\infty)}(A:B)_\rho)
&=\max_{\sigma_A\in \mathcal{S}(A)} \max_{U_A\in \mathcal{U}(A)}
\tr[ \tr_A[\sqrt{\rho_{AB}}\sqrt{\sigma_A} U_A] \tr_A[U_A^\dagger \sqrt{\sigma_A}\sqrt{\rho_{AB}}] ]
\label{eq:cor-max-u-1}\\
&=\max_{\sigma_A\in \mathcal{S}(A)}\tr[\tr_A[\sqrt{\rho_{AB}}\sqrt{\sigma_A}] \tr_A[\sqrt{\sigma_A}\sqrt{\rho_{AB}}]]
=\exp(-I_{1/2}^{\downarrow\downarrow}(A:B)_\rho) .
\label{eq:cor-max-u-2}
\end{align}
\end{cor}
\begin{proof}
Let $\ket{\hat{\rho}}_{ABA^*B^*}$ be the canonical purification of $\rho_{AB}$ 
and let $\hat{\rho}_{AA^*}\coloneqq \tr_{BB^*}[\proj{\hat{\rho}}_{ABA^*B^*}]$. 
The equality in~\eqref{eq:cor-max-u-1} holds because
\begin{align}
\exp(-S_R^{(\infty)}(A:B)_\rho)
&=\max_{\substack{\ket{\phi}\in AA^*:\\ \brak{\phi}=1}} \tr[\hat{\rho}_{AA^*}\proj{\phi}_{AA^*}]
\label{eq:opt2-0}\\
&=\max_{\substack{\ket{\phi}\in AA^*:\\ \brak{\phi}=1}} \bra{\Omega}_{AA^*}\tr_B[\sqrt{\rho_{AB}}\proj{\phi}_{AA^*}\sqrt{\rho_{AB}}] \ket{\Omega}_{AA^*}
\label{eq:opt2-1}\\
&= \max_{\sigma_A\in \mathcal{S}(A)} \max_{U_A\in \mathcal{U}(A)}
\tr[ \tr_A[\sqrt{\rho_{AB}}\sqrt{\sigma_A} U_A] \tr_A[U_A^\dagger \sqrt{\sigma_A}\sqrt{\rho_{AB}}] ] .
\label{eq:opt2}
\end{align}
\eqref{eq:opt2-0} follows from~\eqref{eq:h-inf}. 
\eqref{eq:opt2-1} follows from~\eqref{eq:purification-abab}. 
\eqref{eq:opt2} follows from Proposition~\ref{prop:canonical-invariant}~(c). 

The second equality in~\eqref{eq:cor-max-u-2} follows from the definition of $I_{1/2}^{\downarrow\downarrow}(A:B)_\rho$ due to~\eqref{eq:i-gen-explicit}. 

The first equality in~\eqref{eq:cor-max-u-2} then follows from Theorem~\ref{thm:srinf-i12}.
\end{proof}

Note that the proof of Corollary~\ref{cor:max-u} proceeds in such a way that~\eqref{eq:cor-max-u-1} and the second equality in~\eqref{eq:cor-max-u-2} are derived first. 
The first equality in~\eqref{eq:cor-max-u-2} then follows as a corollary of Theorem~\ref{thm:srinf-i12}. 
Remarkably, the first equality in~\eqref{eq:cor-max-u-2} is non-trivial, 
as the optimization over $U_A$ in~\eqref{eq:cor-max-u-1} cannot generally be solved by simply choosing $U_A=1_A$. 
More precisely, the assertion that for all $\rho_{AB}\in \mathcal{S}(AB),\sigma_A\in \mathcal{S}(A)$
\begin{align}
\max_{U_A\in \mathcal{U}(A)}\tr[\tr_A[\sqrt{\rho_{AB}} \sqrt{\sigma_A} U_A]\tr_A[U_A^\dagger \sqrt{\sigma_A} \sqrt{\rho_{AB}}] ]
&\stackrel{?}{=}\tr\left[ \tr_A[\sqrt{\rho_{AB}}\sqrt{\sigma_A}] \tr_A[\sqrt{\sigma_A}\sqrt{\rho_{AB}}] \right]
\label{eq:sup-u}
\end{align}
is false in general; a counterexample to~\eqref{eq:sup-u} is given below. 
Our proof technique for Theorem~\ref{thm:srinf-i12} circumvents the seemingly difficult task of directly proving the first equality in~\eqref{eq:cor-max-u-2} by instead leveraging the $\mathcal{CF}$-invariance of $\hat{\rho}_{AA^*}$ and employing techniques for $\mathcal{CF}$-invariant objects from Appendix~\ref{app:invariance}. 

\begin{ex}[Counterexample to~\eqref{eq:sup-u}]\label{ex:counter}
Let $A$ and $B$ be $2$-dimensional Hilbert spaces. 
Let $\rho_{AB}\coloneqq \proj{0}_A\otimes\proj{0}_B\in \mathcal{S}(AB)$, 
$\ket{\pm}_A\coloneqq \frac{1}{\sqrt{2}}(\ket{0}_A\pm \ket{1}_A)$, and 
$U_A\coloneqq \ketbra{+}{0}_A+\ketbra{-}{1}_A\in \mathcal{U}(A)$. 
Let $p\in [0,1]$ 
and let $\sigma_A\coloneqq p\proj{+}_A+(1-p)\proj{-}_A\in \mathcal{S}(A)$. 
On the one hand
\begin{align}
x_U
&\coloneqq \tr[\tr_A[\sqrt{\rho_{AB}} \sqrt{\sigma_A} U_A]\tr_A[U_A^\dagger \sqrt{\sigma_A} \sqrt{\rho_{AB}}] ]
\\
&=\lvert\bra{0}\sqrt{\sigma_A} U_A\ket{0}_A\vert^2
=\lvert \sqrt{p}\braket{0}{+}_A\rvert^2 
=\frac{p}{2} .
\end{align}
On the other hand
\begin{align}
x_1
&\coloneqq \tr[\tr_A[\sqrt{\rho_{AB}} \sqrt{\sigma_A}]\tr_A[\sqrt{\sigma_A} \sqrt{\rho_{AB}}] ]
\\
&=\lvert\bra{0}\sqrt{\sigma_A}\ket{0}_A\rvert^2
=\bigg\lvert \frac{\sqrt{p}}{2}+\frac{\sqrt{1-p}}{2} \bigg\rvert^2
=\frac{1}{2}\sqrt{p}\sqrt{1-p}+\frac{1}{4}.
\end{align}
The difference $x_U-x_1$ changes sign at $p_0\coloneqq \frac{2+\sqrt{2}}{4}\approx 0.85$. 
If $p=p_0$, then $x_U=x_1$. 
If $p\in [0,p_0)$, then $x_U<x_1$.
If $p\in (p_0,1]$, then $x_U>x_1$. 
Therefore,~\eqref{eq:sup-u} is violated for any $p\in (p_0,1]$, as the left-hand side of~\eqref{eq:sup-u} is strictly greater than the right-hand side of~\eqref{eq:sup-u}.
\end{ex}

\subsection{Relations between optimizers involved in Theorem~\ref{thm:srinf-i12}}
\label{app:thm-relations}

In the proof of Theorem~\ref{thm:srinf-i12}, we derived the following chain of equalities.
\begin{align}
\exp(-S_R^{(\infty)}(A:B)_\rho )
&= \max_{\substack{\ket{\phi}\in AA^*: \\ \brak{\phi}=1}}
\tr[\hat{\rho}_{AA^*}\proj{\phi}]
\label{eq:op1}\\
&=\max_{\substack{\ket{\psi}\in AA^*: \\ \brak{\psi}=1,\\ (\swapf \ket{\psi})^*=\ket{\psi} }} \tr[\hat{\rho}_{AA^*}\proj{\psi}]
\label{eq:op2}
\\
&=\max_{\substack{\ket{\chi}\in AA^*: \\ \exists\sigma_A\in \mathcal{S}(A):\\ \ket{\chi}=\sqrt{\sigma_A}\ket{\Omega}_{AA^*} }} \tr[\hat{\rho}_{AA^*}\proj{\chi}]
\label{eq:op3}
\\
&=\exp(-\min_{\sigma_A\in \mathcal{S}(A)}I_{1/2}^{\downarrow}(\rho_{AB}\| \sigma_A ))
= \exp(- I_{1/2}^{\downarrow\downarrow}(A:B)_{\rho})
\label{eq:op4}
\end{align}
The following proposition describes how the optimizers for these optimization problems are related to one another.

\begin{prop}[Relations between optimizers]\label{prop:optimizers}
Let $\rho_{AB}\in \mathcal{S}(AB)$. Then all of the following hold.
\begin{enumerate}[label=(\alph*)]
\item If $\sigma_A$ is an optimizer for~\eqref{eq:op4}, then $\ket{\chi}\coloneqq \sqrt{\sigma_A}\ket{\Omega}_{AA^*}$ is an optimizer for~\eqref{eq:op1},~\eqref{eq:op2}, and~\eqref{eq:op3},
and $\sigma_A\ll \rho_A$.
\item If $\ket{\chi}$ is an optimizer for~\eqref{eq:op3}, 
then $\sigma_A\coloneqq \tr_{A^*}[\proj{\chi}]$ is an optimizer for~\eqref{eq:op4}.
\item If $\ket{\psi}$ is an optimizer for~\eqref{eq:op2}, 
then $\ket{\chi}\coloneqq \sqrt{\sigma_A}\ket{\Omega}_{AA^*}$ is an optimizer for~\eqref{eq:op3} where $\sigma_A\coloneqq \tr_{A^*}[\proj{\psi}]$.
\item Suppose $\ket{\phi}$ is an optimizer for~\eqref{eq:op1}. Let
$S_{AA^*}\coloneqq \frac{1}{2}\proj{\phi}+\frac{1}{2}(\proj{\phi})^{\swapf *}$.

If the largest eigenvalue of $S_{AA^*}$ is non-degenerate:
Let $\ket{e}\in AA^*$ be a unit eigenvector of $S_{AA^*}$ associated with the largest eigenvalue of $S_{AA^*}$.
Then there exists $\theta\in [0,2\pi)$ such that $\ket{\psi}\coloneqq e^{i\theta}\ket{e}$
is an optimizer for~\eqref{eq:op2}.

Else: Let $\ket{\psi}\coloneqq \frac{1}{\sqrt{2}}(\ket{\phi}+(\swapf \ket{\phi})^*)$.
Then $\ket{\psi}$
is an optimizer for~\eqref{eq:op2}.
\end{enumerate}
\end{prop}

\begin{proof}[Proof of~(a).]
Let $\sigma_A\in \mathcal{S}(A)$ be an optimizer for~\eqref{eq:op4} and let $\ket{\chi}\coloneqq\sqrt{\sigma_A}\ket{\Omega}_{AA^*}$.
Then, $\ket{\chi}$ is an optimizer for~\eqref{eq:op3} by the same reasoning as in~\eqref{eq:thm-proof0}. 
Since $\ket{\chi}$ is an optimizer for~\eqref{eq:op3}, $\ket{\chi}$ is also an optimizer for~\eqref{eq:op1} and~\eqref{eq:op2} due to Proposition~\ref{prop:canonical-invariant}. 
The assertion that any optimizer for~\eqref{eq:op4} is such that $\sigma_A\ll  \rho_A$ has been proved in~\cite{burri2024doublyminimizedpetzrenyi} (see also~\cite{hayashi2016correlation}).
\end{proof}
\begin{proof}[Proof of~(b)]
The assertion follows from the same reasoning as in~\eqref{eq:thm-proof0}.
\end{proof}
\begin{proof}[Proof of~(c)]
Let $\ket{\psi}\in AA^*$ be an optimizer for~\eqref{eq:op2}. 
Let $\sigma_A\coloneqq\tr_{A^*}[\proj{\psi}]$ and let $\ket{\chi}\coloneqq \sqrt{\sigma_A}\ket{\Omega}_{AA^*}$.
By the same reasoning as in~\eqref{eq:s5}--\eqref{eq:s6}, it follows that $\ket{\chi}$ is an optimizer for~\eqref{eq:op3}.
\end{proof}
\begin{proof}[Proof of~(d).]
The assertion follows from Proposition~\ref{prop:inv-self-adjoint}~(c).
\end{proof}

\subsection{Proof of Corollary~\ref{cor:bounds}}\label{app:proof_cor_bounds}
\begin{proof} 
Let $\ket{\hat{\rho}}_{ABA^*B^*}$ be the canonical purification of $\rho_{AB}$ 
and let $\hat{\rho}_{AA^*}\coloneqq \tr_{BB^*}[\proj{\hat{\rho}}_{ABA^*B^*}]$. 
The first inequality in~\eqref{eq:bounds-1} follows from the bound $\frac{1}{2}H_2(AA^*)_{\hat{\rho}}\leq H_\infty(AA^*)_{\hat{\rho}}$ for R\'enyi entropies. 
The equality in~\eqref{eq:bounds-1} follows from Theorem~\ref{thm:srinf-i12}.
The second inequality in~\eqref{eq:bounds-1} follows from~\eqref{eq:prmi1},~\eqref{eq:prmi2}, and~\eqref{eq:h-alpha}.
The first inequality in~\eqref{eq:bounds-2} holds trivially. 
The second inequality in~\eqref{eq:bounds-2} holds by monotonicity of the singly minimized PRMI in the R\'enyi order and the fact that $I_1^{\uparrow\downarrow}(A:B)_\rho=I(A:B)_\rho$~\cite{hayashi2016correlation,burri2024doublyminimizedpetzrenyi}. 
The last inequality in~\eqref{eq:bounds-2} follows from~\eqref{eq:i-leq-sr}.
\end{proof}

\subsection{Proof of Corollary~\ref{cor:equivalence}}\label{app:proof_cor_equivalence}
\begin{proof}
The inequality $S_R^{(\alpha)}(A:B)_\rho\leq 2\log r_A$ for $\alpha\in [0,\infty]$ follows from~\eqref{eq:sr-bounds}.

It remains to prove the equivalence statement in~(a),  
as the equivalence statements in~(b) and~(c) follow from~(a). 

Let $\alpha\in (0,\infty]$ be arbitrary but fixed. 
Let $\ket{\hat{\rho}}_{ABA^*B^*}$ be the canonical purification of $\rho_{AB}$ 
and let $\hat{\rho}_{AA^*}\coloneqq \tr_{BB^*}[\proj{\hat{\rho}}_{ABA^*B^*}]$. 

First, suppose $S_R^{(\alpha)}(A:B)_\rho=2\log r_A$.
Then also $S_R^{(0)}(A:B)_\rho=2\log r_A$ due to~\eqref{eq:sr-bounds}. 
Hence, $\rank(\hat{\rho}_{AA^*})=r_A^2$ and $H_\alpha(AA^*)_{\hat{\rho}}=H_0(AA^*)_{\hat{\rho}}=2\log r_A$.
Since $\alpha>0$, this implies that $\hat{\rho}_{AA^*}=\hat{\rho}_{AA^*}^0/r_A^2$. 
Hence, $H_\beta(AA^*)_{\hat{\rho}}=2\log r_A$ for all $\beta\in [0,\infty]$. 
In particular, $H_\infty(AA^*)_{\hat{\rho}}=2\log r_A$, so $S_R^{(\infty)}(A:B)_\rho=2\log r_A$. 
This implies that $I(A:B)_\rho=2\log r_A$ due to Corollary~\ref{cor:bounds}.

Now, suppose $I(A:B)_\rho=2\log r_A$ instead. 
This implies that $I_{1/2}^{\downarrow\downarrow}(A:B)_\rho=2\log r_A$~\cite{burri2024doublyminimizedpetzrenyi}.
By Theorem~\ref{thm:srinf-i12}, this is equivalent to $S_R^{(\infty)}(A:B)_\rho=2\log r_A$. 
This implies that $S_R^{(\alpha)}(A:B)_\rho=2\log r_A$ due to the monotonicity of the R\'enyi reflected entropy and its bounds (see Section~\ref{ssec:renyi-reflected}). 
\end{proof}

\subsection{Counterexample for ``Maximality of entanglement of purification implies maximality of mutual information'' (Section~\ref{ssec:comparison_eop})}\label{app:eop}
Let $d_A,d_B\in \mathbb{N}_{\geq 2}$ be such that $d_A\leq d_B$. 
Let $\rho_{AB}\coloneqq d_A^{-1}\sum_{j\in [d_A]}\proj{j}_{A}\otimes\proj{j}_B$. 
The canonical purification of $\rho_{AB}$ is 
\begin{align}
\ket{\hat{\rho}}_{ABA^*B^*}
= d_A^{-1/2}\sum_{j\in [d_A]}\ket{j}_A\otimes\ket{j}_B\otimes\ket{j}_{A^*}\otimes\ket{j}_{B^*}.
\end{align}
Let $\hat{\rho}_{AA^*B^*}\coloneqq \tr_{B}[\proj{\hat{\rho}}_{ABA^*B^*}]$. 
The entanglement of purification is bounded as 
\begin{align}
\log d_A
\geq E_P(A:B)_\rho 
&\geq \min_{\mathcal{M}_{A^*B^*\rightarrow C}\in \CPTP(A^*B^*,C)}H(AC)_{\mathcal{M}_{A^*B^*\rightarrow C}(\hat{\rho}_{AA^*B^*})}
\label{eq:counter-ep0}\\
&=H(A)_{\rho} + \min_{\mathcal{M}_{A^*B^*\rightarrow C}\in \CPTP(A^*B^*,C)}H(C|A)_{\mathcal{M}_{A^*B^*\rightarrow C}(\hat{\rho}_{AA^*B^*})}
\\
&=\log d_A + \min_{\mathcal{M}_{A^*B^*\rightarrow C}\in \CPTP(A^*B^*,C)}H(C|A)_{\mathcal{M}_{A^*B^*\rightarrow C}(\hat{\rho}_{AA^*B^*})}
\label{eq:counter-ep}
\\
&\geq \log d_A
\label{eq:counter-ep2}
\end{align}
where $C$ denotes an arbitrary finite-dimensional Hilbert space, and it is understood that the minimizations also range over the dimension of $C$. 
The first inequality in~\eqref{eq:counter-ep0} follows from~\eqref{eq:eop-bounds}.
\eqref{eq:counter-ep} holds because $\rho_A=1_A/d_A$. 
\eqref{eq:counter-ep2} holds because 
$\hat{\rho}_{AA^*B^*}=d_A^{-1}\sum_{j\in [d_A]}\proj{j}_A\otimes\proj{j}_{A^*}\otimes\proj{j}_{B^*}$ is separable with respect to the bipartition $(A,A^*B^*)$, and because the linearity of $\mathcal{M}_{A^*B^*\rightarrow C}$ implies that $\mathcal{M}_{A^*B^*\rightarrow C}(\hat{\rho}_{AA^*B^*})$ is separable with respect to the bipartition $(A,C)$, 
so $H(C|A)_{\mathcal{M}_{A^*B^*\rightarrow C}(\hat{\rho}_{AA^*B^*})}\geq 0$. 
The inequalities in~\eqref{eq:counter-ep0}--\eqref{eq:counter-ep2} imply that 
\begin{align}
E_P(A:B)_\rho =\log d_A=\log \min\{d_A,d_B\}.
\end{align}

The mutual information of $\rho_{AB}$ is 
\begin{align}
I(A:B)_\rho =\log d_A<2\log d_A=2\log \min\{d_A,d_B\}.
\end{align}

Therefore, maximality of the entanglement of purification does not imply maximality of the mutual information.

\subsection{Examples for bounds on entanglement of purification (Section~\ref{ssec:comparison_eop})}\label{app:eop2}
For the first example, let $d_A\in \mathbb{N}_{\geq 2}$ and 
$\ket{\Phi}\coloneqq d_A^{-1/2}\sum_{j\in [d_A]}\ket{j}_A\otimes\ket{j}_{A^*}$. 
Then, 
\begin{align}
E_P(A:A^*)_{\proj{\Phi}}=\log d_A<2\log d_A = S_R^{(\infty)}(A:A^*)_{\proj{\Phi}}.
\end{align} 

For the second example, let $d\in \mathbb{N}_{\geq 2}$, 
let $A$ and $B$ be $d$-dimensional Hilbert spaces, 
and let $A'$ be a $2$-dimensional Hilbert space. 
Let $\omega_{A'AB}\in \mathcal{S}(A'AB)$ be defined as in~\cite[Corollary 10]{christandl2005uncertainty}, i.e., 
\begin{align}\label{eq:omega_def}
\omega_{A'AB}\coloneqq \frac{1}{d^2}\proj{0}_{A'}\otimes \frac{1}{2}(1_{AB}+\mathcal{F}_{AB} )
+\frac{1}{d^2}\proj{1}_{A'}\otimes \frac{1}{2}(1_{AB}-\mathcal{F}_{AB} ).
\end{align}
Then, $E_P(A'A:B)_\omega =\log d$~\cite[Corollary 10]{christandl2005uncertainty}. The definition in~\eqref{eq:omega_def} implies that
\begin{align}
\omega_{B}
&=1_B/d,\\
\omega_{A'A}
&=\frac{d+1}{2d^2}\proj{0}_{A'}\otimes 1_A+\frac{d-1}{2d^2}\proj{1}_{A'}\otimes 1_A.
\end{align}
Hence, 
\begin{align}
I(A'A:B)_\omega 
&=H(A'A)_\omega + H(B)_\omega - H(A'AB)_\omega
\\
&=\frac{d+1}{2d}\log\left(\frac{2d^2}{d+1}\right) 
+\frac{d-1}{2d}\log\left(\frac{2d^2}{d-1}\right)
+\log d -\log d^2
\\
&<\log d 
=E_P(A'A:B)_\omega.
\end{align}

Together, these two examples imply that the entanglement of purification cannot, in general, be inserted anywhere in the chain of inequalities presented in Corollary~\ref{cor:bounds}.

\subsection{Proof of Theorem~\ref{thm:minimized}}\label{app:proof_minimized}
\begin{proof}
By~\eqref{eq:sr-ineq}, $S_R^{\downarrow (n)}(A:B)_\rho \leq \inf_{s\in \mathbb{R}} S_{D,s}^{(n)}(A:B)_\rho\leq S_{D,0}^{(n)}(A:B)_\rho=S_R^{(n)}(A:B)_\rho$. 
It remains to prove that 
$S_R^{\downarrow (n)}(A:B)_\rho \geq S_R^{(n)}(A:B)_\rho$.

\emph{Case 1: $n\in \mathbb{N}_{\geq 2}$.} 
Let $\ket{\hat{\rho}}_{ABA^*B^*}\coloneqq \sqrt{\rho_{AB}}\ket{\Omega}_{AA^*}\otimes\ket{\Omega}_{BB^*}$ 
and let $\hat{\rho}_{AA^*}\coloneqq \tr_{BB^*}[\proj{\hat{\rho}}_{ABA^*B^*}]$. 
Then
\begin{align}\label{eq:proof-min1}
S_R^{(n)}(A:B)_\rho = H_n(AA^*)_{\hat{\rho}} =\frac{1}{1-n}\log\tr[\hat{\rho}_{AA^*}^n]  .
\end{align}
Let $U_{AB}\in \mathcal{U}(AB)$ be such that $[\rho_{AB},U_{AB}]=0$. 
Let $\ket{\bar{\rho}}_{ABA^*B^*}\coloneqq \sqrt{\rho_{AB}} U_{AB}\ket{\Omega}_{AA^*}\otimes\ket{\Omega}_{BB^*}$ 
and let $\bar{\rho}_{AA^*}\coloneqq \tr_{BB^*}[\proj{\bar{\rho}}_{ABA^*B^*}]$. 
Then
\begin{align}\label{eq:proof-min2}
H_n(AA^*)_{\bar{\rho}}
= \frac{1}{1-n}\log \tr[\bar{\rho}^n_{AA^*}] .
\end{align}
By Proposition~\ref{prop:invariant}~(d) and~(e), there exist an orthonormal basis $\{\ket{e_j}\}_{j\in [d_Ad_B]}$ for $AB$, 
$(p_j)_{j\in [d_Ad_B]}\in [0,1]^{\times d_Ad_B}$, 
and $(\theta_j)_{j\in [d_Ad_B]}\in [0,2\pi)^{\times d_Ad_B}$ such that
\begin{align}\label{eq:proof-min21}
\ket{\bar{\rho}}_{ABA^*B^*}
=\sum_{j\in [d_A d_B]}e^{i\theta_j}\sqrt{p_j}\ket{e_j}_{AB}\otimes\ket{e_j}_{A^*B^*}^* .
\end{align}
Then, $\rho_{AB}=\tr_{A^*B^*}[\proj{\bar{\rho}}_{ABA^*B^*}]=\sum_{j\in [d_A d_B]} p_j\proj{e_j}_{AB}$, so 
\begin{align}\label{eq:proof-min22}
\ket{\hat{\rho}}_{ABA^*B^*}
=\sum_{j\in [d_A d_B]}\sqrt{p_j}\ket{e_j}_{AB}\otimes\ket{e_j}_{A^*B^*}^* .
\end{align}
Let us define the following two linear operators on $A$ for any $j_1,\dots, j_n,k_1,\dots, k_n\in [d_Ad_B]$.
\begin{align}
X^{(j_1,k_1,\dots,j_n, k_n)}
&\coloneqq \tr_B[\ketbra{e_{j_1}}{e_{k_1}}_{AB}]
\quad\cdots\quad 
\tr_B[\ketbra{e_{j_n}}{e_{k_n}}_{AB}]\\
Y^{(j_1,k_1,\dots,j_n, k_n)}
&\coloneqq\tr_B[(\ketbra{e_{j_1}}{e_{k_1}}_{AB})^*]
\quad\cdots\quad 
\tr_B[(\ketbra{e_{j_n}}{e_{k_n}}_{AB})^*]
\end{align}
Then, 
\begin{align}
\tr [\bar{\rho}_{AA^*}^n]
&=\sum_{\substack{j_1,\dots, j_n\in [d_Ad_B],\\ k_1,\dots, k_n\in [d_Ad_B]}}
\left(\prod_{l=1}^n\sqrt{p_{j_l}p_{k_l}} e^{i(\theta_{j_l}-\theta_{k_l})}\right)
\tr[X^{(j_1,k_1,\dots,j_n, k_n)}]\tr[Y^{(j_1,k_1,\dots,j_n, k_n)}]
\label{eq:proof-min30}\\
&=\sum_{\substack{j_1,\dots, j_n\in [d_Ad_B],\\ k_1,\dots, k_n\in [d_Ad_B]}}
\left(\prod_{l=1}^n\sqrt{p_{j_l}p_{k_l}} e^{i(\theta_{j_l}-\theta_{k_l})}\right)
\lvert \tr[X^{(j_1,k_1,\dots,j_n, k_n)}]\rvert^2
\label{eq:proof-min3}\\
&\leq \sum_{\substack{j_1,\dots, j_n\in [d_Ad_B],\\ k_1,\dots, k_n\in [d_Ad_B]}}
\Big\lvert \left(\prod_{l=1}^n\sqrt{p_{j_l}p_{k_l}} e^{i(\theta_{j_l}-\theta_{k_l})}\right)
\lvert \tr[X^{(j_1,k_1,\dots,j_n, k_n)}]\rvert^2\Big\rvert \\
&= \sum_{\substack{j_1,\dots, j_n\in [d_Ad_B],\\ k_1,\dots, k_n\in [d_Ad_B]}}
\left(\prod_{l=1}^n\sqrt{p_{j_l}p_{k_l}} \right)
\lvert\tr[X^{(j_1,k_1,\dots,j_n, k_n)}]\rvert^2
=\tr [\hat{\rho}_{AA^*}^n] .
\label{eq:proof-min4}
\end{align}
\eqref{eq:proof-min30} follows from~\eqref{eq:proof-min21}. 
\eqref{eq:proof-min3} follows from $(X^{(j_1,k_1,\dots,j_n, k_n)})^*=Y^{(j_1,k_1,\dots,j_n, k_n)}$. 
\eqref{eq:proof-min4} follows from~\eqref{eq:proof-min22}. 
By~\eqref{eq:proof-min1} and~\eqref{eq:proof-min2}, we can conclude that 
\begin{align}\label{eq:proof-min5}
H_n(AA^*)_{\bar{\rho}}\geq S_R^{(n)}(A:B)_\rho .
\end{align}
Since the given argument is valid for every $U_{AB}\in \mathcal{U}(AB)$ such that $[\rho_{AB},U_{AB}]=0$, it follows from~\eqref{eq:def-m2} and~\eqref{eq:proof-min5} that 
$S_R^{\downarrow (n)}(A:B)_\rho\geq S_R^{(n)}(A:B)_\rho$.

\emph{Case 2: $n=\infty$.} 
Let $U_{AB}\in \mathcal{U}(AB)$ be such that $[\rho_{AB},U_{AB}]=0$. 
Let $\ket{\bar{\rho}}_{ABA^*B^*}\coloneqq \sqrt{\rho_{AB}} U_{AB}\ket{\Omega}_{AA^*}\otimes\ket{\Omega}_{BB^*}$ 
and let $\bar{\rho}_{AA^*}\coloneqq \tr_{BB^*}[\proj{\bar{\rho}}_{ABA^*B^*}]$. 
Case~1 implies that for all $m\in \mathbb{N}_{\geq 2}$
\begin{align}\label{eq:proof_hm}
H_m(AA^*)_{\bar{\rho}}
\geq S_R^{\downarrow (m)}(A:B)_\rho 
\geq S_R^{(m)}(A:B)_\rho .
\end{align}
By taking the limit $m\rightarrow\infty$, it follows that
\begin{align}
H_\infty(AA^*)_{\bar{\rho}}
=\lim_{m\rightarrow\infty}H_m(AA^*)_{\bar{\rho}}
\stackrel{\eqref{eq:proof_hm}}{\geq} \lim_{m\rightarrow\infty} S_R^{(m)}(A:B)_\rho
= S_R^{(\infty)}(A:B)_\rho .
\end{align}
By~\eqref{eq:def-m2}, we can conclude that 
$S_R^{\downarrow (\infty)}(A:B)_\rho\geq S_R^{(\infty)}(A:B)_\rho$.
\end{proof}

\subsection{Example for Remark~\ref{rem:02}}\label{app:ex-02}

Let $A$ and $B$ be Hilbert spaces of dimension $d_A=d_B=3$. Let $p\coloneqq \frac{1}{6}$, $s\in \{0,1\}$, and let us define the following matrices.
\begin{align}
P&\equiv (P_{j,k})_{\substack{j\in [d_A], k\in [d_B]}} \coloneqq \begin{pmatrix}
p&0&1/2-2p\\
0&p&1/2-2p\\
p&p&0\\
\end{pmatrix}\\
M&\equiv (M_{j,k})_{\substack{j\in [d_A], k\in [d_B]}} 
\coloneqq (\sqrt{P_{j,k}} (-1)^{s\delta_{j,0}\delta_{k,0}})_{\substack{j\in [d_A], k\in [d_B]}} 
= \begin{pmatrix}
(-1)^s\sqrt{p}&0&\sqrt{1/2-2p}\\
0&\sqrt{p}&\sqrt{1/2-2p}\\
\sqrt{p}&\sqrt{p}&0\\
\end{pmatrix}
\end{align}
Furthermore, we define the following objects.
\begin{align}
\rho_{AB}&\coloneqq \sum_{j\in [d_A],k\in [d_B]}P_{j,k}\proj{j}_{A} \otimes\proj{k}_{B} 
\in \mathcal{S}(AB)
\label{eq:ex-rho-def}\\
U_{AB}&\coloneqq \sum_{j\in [d_A],k\in [d_B]}(-1)^{\delta_{j,0}\delta_{k,0} }\proj{j}_{A} \otimes\proj{k}_{B} 
\in \mathcal{U}(AB)
\\
\ket{\hat{\rho}}_{ABA^*B^*}
&\coloneqq \sqrt{\rho_{AB}} \ket{\Omega}_{AA^*}\otimes \ket{\Omega}_{BB^*}
\\\ket{\bar{\rho}}_{ABA^*B^*}
&\coloneqq \sqrt{\rho_{AB}} U_{AB}\ket{\Omega}_{AA^*}\otimes \ket{\Omega}_{BB^*}
\end{align}
Let $\hat{\rho}_{AA^*}\coloneqq \tr_{BB^*}[\proj{\hat{\rho}}_{ABA^*B^*}]$ 
and let $\bar{\rho}_{AA^*}\coloneqq \tr_{BB^*}[\proj{\bar{\rho}}_{ABA^*B^*}]$. 
Then,
\begin{align}
\sum_{j,k\in [d_A]} (MM^T)_{j,k}\ketbra{j}{k}_{A}\otimes \ketbra{j}{k}_{A^*}
=\begin{cases}
\hat{\rho}_{AA^*}
\qquad\text{if }s=0\\
\bar{\rho}_{AA^*}
\qquad\text{if }s=1
\end{cases}
\end{align}
where $M^T$ denotes the transpose of $M$. 
The spectrum of $MM^T$ is given by 
$(\frac{2}{3},\frac{1}{6},\frac{1}{6})$ if $s=0$ 
and by $(\frac{1}{2},\frac{1}{2},0)$ if $s=1$.
Hence, for any $n\in [0,1)\cup(1,\infty)$
\begin{align}
S_R^{(n)}(A:B)_\rho 
&=H_n (AA^*)_{\hat{\rho}}
=\frac{1}{1-n}\log \tr[\hat{\rho}_{AA^*}^n]
= \frac{1}{1-n}\log \left(2\left(\frac{1}{6}\right)^n +\left(\frac{2}{3}\right)^n \right), 
\label{eq:sr-down-sr1}\\
S_R^{\downarrow (n)}(A:B)_\rho 
&\leq H_n (AA^*)_{\bar{\rho}}
=\frac{1}{1-n}\log \tr[\bar{\rho}_{AA^*}^n]
=\frac{1}{1-n}\log \left(2\left(\frac{1}{2}\right)^n \right)
=\log 2 .
\label{eq:sr-down-sr2}
\end{align}
The inequality in~\eqref{eq:sr-down-sr2} holds because $[\rho_{AB},U_{AB}]=0$, see~\eqref{eq:def-m2}.

\begin{figure}
\begin{tikzpicture}
\begin{axis}[
    axis lines = left,
    xlabel = \(n\),
    ymin = 0.5,
    ymax = 1.2,
    height = 0.3\textwidth,
    width = 0.5\textwidth,
    legend pos=north east,
    clip=false
]

\addplot[color=black,thick]table[x=alpha,y=sr]{data_sr.txt};
\addlegendentry{\(S_R^{(n)}(A:B)_\rho\)}

\addplot [domain=0:3,color=black,thick,dashed]{ln(2)};
\addlegendentry{\(\log 2\)}

\end{axis}
\end{tikzpicture}
\caption{The plot depicts the expressions on the right-hand sides of~\eqref{eq:sr-down-sr1} and~\eqref{eq:sr-down-sr2}. 
By~\eqref{eq:sr-down-sr2} and~\eqref{eq:sr-ineq}, both lines are an upper bound on $S_R^{\downarrow (n)}(A:B)_\rho$. 
Since the dashed line lies strictly below the solid line for all $n\in [0,2)$, it follows that 
$S_R^{\downarrow (n)}(A:B)_\rho<S_R^{(n)}(A:B)_\rho$ for all $n\in [0,2)$.
}
\label{fig:sr}
\end{figure}
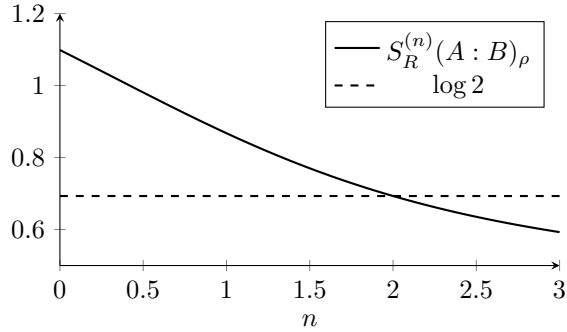

The expressions on the right-hand sides of~\eqref{eq:sr-down-sr1} and~\eqref{eq:sr-down-sr2} are plotted in Figure~\ref{fig:sr}. 
As can be seen from this figure, 
$S_R^{\downarrow (n)}(A:B)_\rho <S_R^{(n)}(A:B)_\rho$
for all $n\in [0,2)$.

\bibliographystyle{arxiv_fullname}
\bibliography{bibfile}

\end{document}